\documentclass[11pt]{article}
 
\usepackage[margin=1in]{geometry} 
\usepackage{amsmath,amsthm,amssymb, color}
\usepackage{enumerate} 
\usepackage[colorlinks=true, linkcolor=blue]{hyperref}
\usepackage{color}
\usepackage{multirow}
\usepackage{enumitem}
\usepackage{mathtools}
\usepackage[normalem]{ulem}

\DeclarePairedDelimiter{\ceil}{\lceil}{\rceil}

\newcommand{\OR}{\textnormal{OR}}
\newcommand{\AND}{\textnormal{AND}}

\newcommand{\THR}{\textnormal{THR}}
\newcommand{\DIST}{\textnormal{DIST}}
\newcommand{\exact}{\textnormal{EXACT}}

\newcommand{\N}{\mathbb{N}}

\newcommand{\R}{\mathbb{R}}
\newcommand{\E}{\mathbb{E}}
\newcommand{\W}{\mathcal{W}}
\newcommand{\bra}[1]{\left\{#1\right\}}
\newcommand{\mathify}[1]{\ifmmode{#1}\else\mbox{$#1$}\fi}
\newcommand{\abs}[1]{\mathify{\left| #1 \right|}}
\newcommand{\sgn}{\mathrm{sgn}}
\newcommand{\pmone}{\bra{-1, 1}}
\newcommand{\zone}{\bra{0, 1}}

\renewcommand{\deg}{\mathrm{deg}}
\newcommand{\adeg}{\widetilde{\mathrm{deg}}}

\newcommand{\phd}{\textnormal{phd}}

\newtheorem{theorem}{Theorem}[section]
\newtheorem{corollary}[theorem]{Corollary}
\newtheorem{remark}[theorem]{Remark}
\newtheorem{lemma}[theorem]{Lemma}
\newtheorem{claim}[theorem]{Claim}
\newtheorem{fact}[theorem]{Fact}
\newtheorem{definition}[theorem]{Definition}

\newcommand\norm[1]{\left\lVert#1\right\rVert}

\title{Improved Approximate Degree Bounds For $k$-distinctness}
\author{
Nikhil S.~Mande\\
Georgetown University\\
\textsf{nikhil.mande@georgetown.edu}
\and
Justin Thaler\\
Georgetown University\\
\textsf{justin.thaler@georgetown.edu}
\and 
Shuchen Zhu\\
Georgetown University\\
\textsf{sz424@georgetown.edu}
}
\date{}

\begin{document}
\maketitle

\begin{abstract}
An open problem that is widely regarded
as one of the most important
in quantum query complexity is to resolve
the quantum query complexity of the $k$-distinctness
function on inputs of size $N$. While the case
of $k=2$ (also called Element Distinctness) is well-understood,
there is a polynomial gap between the known upper and lower bounds
for all constants $k>2$.
Specifically, the best known upper bound is 
$O\left(N^{(3/4)-1/(2^{k+2}-4)}\right)$ (Belovs, FOCS 2012),
while
the best known lower bound for $k\geq 2$ is
 $\tilde{\Omega}\left(N^{2/3} + N^{(3/4)-1/(2k)}\right)$ 
(Aaronson and Shi, J.~ACM 2004; Bun, Kothari, and Thaler, STOC 2018).

For any constant $k \geq 4$, we improve the lower bound to 
$\tilde{\Omega}\left(N^{(3/4)-1/(4k)}\right)$.
This yields, for example, the first proof that $4$-distinctness
is strictly harder than Element Distinctness. Our lower bound applies more generally to approximate degree. 

As a secondary result, we give a simple construction of an approximating polynomial of degree $\tilde{O}(N^{3/4})$ that applies whenever $k \leq \text{polylog}(N)$. 

\end{abstract}

\section{Introduction}

In quantum query complexity, a quantum algorithm is given query access to
the bits of an unknown input $x$, and the goal is to compute some (known) function $f$ of $x$ while minimizing the number
of bits of $x$ that are queried. In contrast
to classical query complexity,
quantum query algorithms are allowed to make queries in superposition, and the algorithm is not charged for performing unitary operations that are independent of $x$. 
Quantum query complexity is a rich model that allows
for the design of highly sophisticated algorithms and captures much of the power of quantum computing. Indeed, most quantum algorithms were discovered in or can easily be described in the query setting.

An open problem that is widely regarded
as one of the most important
in quantum query complexity \cite{zhandrymulticollisions} is to resolve
the complexity of the \emph{$k$-distinctness} function.
For this function, the input $x$ specifies a list
of $N$ numbers from a given range of size $R$,\footnote{For purposes of this introduction, $N$ and $R$ are assumed to be of the same order of magnitude (up to a factor depending on $k$ alone). For simplicity throughout this section, we state
our bounds purely in terms of $N$, leaving
unstated the assumption that $R$ and $N$ are of the same order
of magnitude.} and the function evaluates to TRUE\footnote{Throughout
this manuscript, we associate $-1$ with logical TRUE and $+1$ with logical FALSE.} if 
there is any range item that appears
$k$ or more times in the list. 
The case $k=2$ corresponds to the complement of the widely-studied \emph{Element Distinctness} function, whose complexity is known to be $\Theta(N^{2/3})$ \cite{ambainis, aaronsonshi}. 

For general values of $k$, the best known upper bound
on the quantum query complexity of $k$-distinctness
is $O\left(N^{3/4-1/(2^{k+2}-4)}\right)$, due
to a highly sophisticated
algorithm of Belovs \cite{belovs}. Belovs' algorithm is based 
on the so-called learning graph framework
in quantum algorithm design, and
improves
over an earlier upper bound
of $O(N^{k/(k+1)})$ due to Ambainis \cite{ambainis}
that is based on quantum walks over the Johnson graph.

For a long time, the best known lower bound
on the quantum query complexity
of $k$-distinctness was $\Omega(N^{2/3})$ for any
$k \geq 2$, due to Aaronson and Shi \cite{aaronsonshi}, with refinements given by Kutin \cite{kutin} and Ambainis \cite{ambainissmallrange}. This lower bound is tight for $k=2$ (matching Ambainis' upper bound \cite{ambainis}), but it is not known to be tight for any $k>2$. Recently, Bun, Kothari, and Thaler~\cite{BKT18} proved a lower bound
of $\tilde{\Omega}(N^{3/4-1/(2k)})$ for constant $k$.\footnote{Throughout this manuscript, $\tilde{O}$, $\tilde{\Omega}$ and $\tilde{\Theta}$ notations are used to hide factors
that are polylogarithmic in $N$.}
This improved over the prior lower bound
of $\Omega(N^{2/3})$ for any constant $k \geq 7$.
Furthermore, combined with Belovs' upper bound, this established that for sufficiently large constants $k$,
the exponent in the quantum query complexity of $k$-distinctness
approaches ${3/4}$ from below. However, the precise
rate at which the quantum query complexity
approaches $N^{3/4}$ remains open: there is a polynomial gap between the upper and lower bounds 
for any constant $k$, and indeed there is a qualitative difference between the inverse-exponential
dependence on $k$ in the exponent of 
$N^{3/4-1/(2^{k+2}-4)}$ (the known upper bound), and
the inverse-linear dependence in the known lower bound of $N^{3/4-1/(2k)}$. 

\paragraph{Main Result.}  In this paper, our main result improves the lower bound from $\tilde{\Omega}(N^{3/4-1/(2k)})$
to $\tilde{\Omega}(N^{3/4-1/(4k)})$. While this bound
is qualitatively similar to the lower bound
of \cite{BKT18}, it offers a polynomial improvement
for every constant $k\geq 4$.  Perhaps more significantly, for $k \in \{4, 5, 6\}$,
it is the first improvement over Aaronson and Shi's
$\Omega(N^{2/3})$ lower bound that has stood for
nearly 20 years.

\paragraph{Approximate Degree.} The \emph{$\epsilon$-error approximate degree}
of a Boolean function $f \colon \{-1, 1\}^n
\to \{-1, 1\}$,
denoted $\adeg_\epsilon(f)$, is the least degree of a real polynomial $p$ such that $|p(x) - f(x)| \leq \epsilon$ for all $x \in \{-1, 1\}^n$. The standard setting
of the error parameter is $\epsilon=1/3$,
and the $(1/3)$-approximate degree of $f$ is denoted
$\adeg(f)$ for brevity. 

As famously observed by Beals et al.~\cite{bealsetal},  the quantum query complexity of
a function $f$ is lower bounded by (one half times) the approximate degree of $f$.
Hence, any lower bound on the approximate degree
of $f$ implies that (up to a factor of 2) the 
same lower bound holds for the quantum query complexity of $f$.

As with prior lower bounds for $k$-distinctness \cite{aaronsonshi, kutin, ambainissmallrange, BKT18}, our $k$-distinctness lower bound is in fact
an approximate degree lower bound (on the natural Boolean function induced by $k$-distinctness on $N \lceil \log_2 R \rceil$ bits, where $R$ denotes the size of the range). Our analysis is a substantial refinement of the lower bound analysis of Bun et al.~\cite{BKT18}.

\begin{theorem}[Informal version of Theorem \ref{thm: main} and Corollary \ref{cor: ambrangereduction}] \label{thm:informal}
For any constant $k \geq 2$, the approximate degree and quantum query complexity of the $k$-distinctness function with domain size $N$ and range size $R\geq N$ is $\tilde{\Omega}(N^{3/4 - 1/(4k)})$. 
\end{theorem}

\begin{remark} Theorem~\ref{thm:informal} provides an approximate degree lower bound for constant error $\epsilon=1/3$. A recent result of Sherstov and Thaler \cite[Theorem 3.4]{vanishingerror}
transforms any constant-error approximate degree lower
bound for $k$-distinctness, into a lower bound for vanishing error $\epsilon=o(1)$. Specifically, combining Theorem \ref{thm:informal} and \cite[Theorem 3.4]{vanishingerror} yields that for constant $k$, the $\epsilon$-error approximate degree of $k$-distinctness is at least $\tilde\Omega\left(N^{3/4-1/(4k)}\log^{1/4+1/(4k)}(1/\epsilon)\right)$, for all $\epsilon\in[(1/3)^N,1/3]$.
\end{remark}

\paragraph{A Secondary Result: The Approximate
Degree for Super-Constant
Values of $k$.}

Recall that for constant $k$, the best known 
approximate degree upper bound for $k$-distinctness, due to Belovs, is
$O\left(N^{3/4-1/(2^{k+2}-4)}\right)$.
For non-constant values of $k$,
the upper bound implied by Belovs' algorithm 
grows exponentially with $k$. That is,
the Big-Oh notation in the 
upper bound hides a leading factor of at least $2^{c k}$ for some positive constant $c$.\footnote{
Belovs' approximate degree upper bound 
was recently reproved by Sherstov \cite{algopoly}, who made the exponential dependence on $k$ explicit (see, e.g., \cite[Theorem 6.6]{algopoly}). To clarify, \label{footnote} Belovs' result is in fact
a quantum query upper bound, which in turn implies an approximate degree upper bound. 
Sherstov's proof  avoids quantum algorithms, and hence
does not yield a quantum query upper bound.} Consequently
Belovs' result is $N^{3/4 + \Omega(1)}$
for any $k \geq \Omega(\log N)$.  Furthermore,
the bound becomes vacuous (i.e., linear in $N$)
for $k \geq c \log N$ for a large enough constant $c>0$. 

Our secondary result improves this state of affairs by giving
a $\tilde{O}(N^{3/4})$
approximate degree upper bound
that holds for any value of $k$
that grows at most polylogarithmically with $N$. 
\begin{theorem}[Informal] \label{thm:informalupper}
For any $k \leq \textnormal{polylog}(N)$,
the approximate degree of $k$-distinctness
is $\tilde{O}(N^{3/4})$. 
\end{theorem}

We mention that for \emph{any} $k \geq 2$,
the approximating polynomials for $k$-distinctness
that follow from prior works \cite{ambainis, belovs, algopoly} are quite complicated, and in our opinion there has not been a genuinely simple construction of any $O(N^{3/4})$-degree
approximating polynomials recorded in the literature, even for the case of $k=2$ (i.e., Element Distinctness).
Accordingly, we feel that Theorem \ref{thm:informalupper} has didactic value even for constant values of $k$ (though the $\tilde{O}(N^{3/4})$ approximate degree upper bound that it achieves is not tight for any constant $k \geq 2$). 

To clarify, Theorem \ref{thm:informalupper}
does \emph{not} yield a quantum query upper
bound, but only an approximate degree upper bound.
Indeed, it remains an interesting open question
whether the quantum query complexity of $k$-distinctness is 
sublinear in $N$ for all $k=\text{polylog}(N)$ (see Section \ref{s:discussion} for further discussion).

Our proof of Theorem \ref{thm:informalupper} is a simple extension of a result of Sherstov \cite[Theorem 1.3]{algopoly} that yielded an $O(N^{3/4})$ approximate degree upper bound for a different function called 
Surjectivity.\footnote{Surjectivity is the function that interprets its input as a list of $N$ numbers from a given range of size $R$, and evaluates to TRUE if and only if every range element appears at least once in the list.}  In Section \ref{s:proofoverviewupper} below, we explain the main observations necessary 
to obtain Theorem \ref{thm:informalupper} via the technique used to prove the upper bound for Surjectivity.

\subsection{Discussion and Open Problems}
\label{s:discussion}
The most obvious and important open question is
to finish resolving the approximate degree
and quantum query complexity of $k$-distinctness
for any $k > 2$. 
Currently, the upper and lower bounds
qualitatively differ in their dependence on $k$,
with the upper bound having an exponent of
the form $3/4-\exp(-O(k))$ and the lower bound
having an exponent of the from $3/4-\Omega(1/k)$.
It seems very likely that major new techniques
will be needed to qualitatively
change the form of \emph{either} the upper or lower bound. In particular, on the lower bounds side, our analysis
is based on a variant of a technique called
 \emph{dual block composition} (see Section \ref{lowerboundoverview}), and
 we suspect that we
have reached the limit of what is provable 
for $k$-distinctness using this technique
and its variants.

We remark here that Liu and Zhandry \cite{zhandrymulticollisions}
recently showed that the quantum query complexity
of a certain \emph{search} version of $k$-distinctness
(defined over randomly generated inputs) is $\Theta(n^{1/2-1/(2^{k}-1)})$.
This inverse-exponential dependence on $k$
is tantalizingly reminsicent of Belovs' upper bound for $k$-distinctness. This may be construed as mild evidence
that $3/4-\exp(-O(k))$ is the right qualitative
bound for $k$-distinctness itself.

A very interesting intermediate goal 
is to establish any polynomial improvement
over the long-standing $\Omega(n^{2/3})$ lower bound
for $3$-distinctness. This would finally establish that $3$-distinctness is strictly harder than Element Distinctness (such a result is now
known for all $k\geq 4$ due to Theorem \ref{thm:informal}).

It would also be interesting 
to resolve the quantum query complexity
of $k$-distinctness for $k = \text{polylog}(N)$.
Although this question my appear to be 
of specialized interest,
we believe that resolving it could
shed light on 
the relationship between approximate degree
and quantum query complexity.
Indeed, while any quantum algorithm for a function $f$
can be turned into an approximating polynomial for $f$ via the transformation of Beals et al.~\cite{bealsetal}, no transformation
in the reverse direction is possible in general \cite{ambainisseparation}. This can be seen, for example, because the quantum query complexity of Surjectivity is known to be $\Omega(N)$~\cite{BM12, She15}, but its approximate degree is $O(N^{3/4})$~\cite{algopoly, BKT18}.  Nonetheless,
approximate degree and quantum query complexity
turn out to coincide for most functions that arise naturally (Surjectivity remains the only function that exhibits a separation, without having been  specifically constructed for that purpose). In our opinion, this phenomenon remains mysterious, and it would be interesting to demystify it.
For example, could one identify special properties of approximating polynomials that would permit a reverse-Beals-et-al.~transformation to turn that polynomial into a quantum query algorithm?\footnote{There are works in this general direction, notably \cite{completelybounded}, which shows that a certain technical refinement of approximate degree, called approximation by completely bounded forms, characterizes quantum query complexity. But to our knowledge these works
have not yielded any novel quantum query upper bounds for any specific function.} 
 Perhaps an $\tilde{O}(N^{3/4})$
upper bound for $(\text{polylog}(N))$-distinctness
could be derived in this manner.
On the other hand, due to our Theorem \ref{thm:informalupper}, any $N^{3/4 + \Omega(1)}$ \emph{lower bound} for $(\text{polylog}(N))$-distinctness would require moving beyond the polynomial method.\footnote{We remark that the positive-weights adversary method
is also incapable of proving such a result due to the certificate complexity barrier.}

\subsection{Paper Roadmap}
We give a high-level overview of the proofs of our lower bound and upper bound in Sections~\ref{lowerboundoverview} and~\ref{s:proofoverviewupper}, respectively.
Section \ref{sec:prelims} covers preliminaries.
The proof of our main theorem (Theorem \ref{thm:informal})
is spread over Sections \ref{sec:detailedoutline}-\ref{sec:endofmainproof}.
Section \ref{sec:detailedoutline}
gives a detailed, technical outline of the proof,
Section \ref{sec:startofmainproof} establishes
some auxiliary lemmas, and Section \ref{sec:endofmainproof}
contains the heart of the proof. 
Finally, Section \ref{sec:upperbound} proves Theorem
\ref{thm:informalupper}. 

\section{Overview of the Proofs}
In this section we give an overview of the proofs of our lower bound and upper bound.

\subsection{The Lower Bound}
\label{lowerboundoverview}

Throughout this subsection we assume that $k \geq 2$ is an arbitrary but fixed constant. 

Let $\THR_N^k$ denote the function on $N$-bit inputs
that evaluates to $-1$ on inputs of Hamming weight at least $k$, and evaluates to $1$ otherwise. 
For $N \leq n$,
let $(\{-1,1\}^n)^{\leq N}$ denote the subset
of $\{-1,1\}^n$ consisting of all
 inputs of Hamming weight at most $N$.
For any function $f_n \colon \{-1,1\}^n \to \{-1,1\}$,\footnote{Throughout, we use subscripts 
where appropriate to
clarify the number of bits over which a function is defined.}
let $f_n^{\leq N}$ denote the partial function
obtained by restricting the domain of $f$ to $(\{-1,1\}^n)^{\leq N}$,
and let $\adeg(f_n^{\leq N})$ denote the least
degree of a real polynomial $p$
such that $|p(x)-f_n(x)| \leq 1/3$ for
all $x \in (\{-1,1\}^n)^{\leq N}$.

Simplifying very slightly, prior work by Bun and Thaler \cite{BT17} 
(building on an important lemma of Ambainis \cite{ambainissmallrange})
implied that for $k\geq 2$ the approximate degree of $k$-distinctness 
is equivalent to $\adeg(f_{R N}^{\leq N})$ for
$f=\OR_R \circ \THR^k_N$. Here, $g_n \circ h_m$ 
denotes the function on $n \cdot m$ bits
obtained by block-composing $g$ and $h$, i.e., 
$g \circ h$ evaluates $h$ on $n$ disjoint inputs and
feeding the outputs of all $n$ copies of $h$ into $g$.

Bun et al.~\cite{BKT18} proved their $\tilde{\Omega}(N^{3/4-1/(2k)})$ lower bound
for $\adeg(f_{R N}^{\leq N})$ via the \emph{method of dual
polynomials}. This is a technique for proving approximate degree lower bounds that works by constructing an explicit solution
to a certain linear program capturing 
the approximate degree of any function.
Specifically, a dual witness to the fact
that $\adeg(f_{RN}^{\leq N}) \geq d$ 
is a function $\psi \colon \{-1,1\}^{RN} \to \R$ satisfying
the following properties. 

First, 
$\psi$ must be uncorrelated with all polynomials $p$
of degree at most $d$, i.e., $\langle \psi, p\rangle = 0$ for all such polynomials $p$, where $\langle \psi, p\rangle = \sum_{x \in\{-1,1\}^{RN} } \psi(x) p(x).$ Such a $\psi$ is said to have \emph{pure high degree} at least $d$.

Second, $\psi$ must be well-correlated with $f$, i.e., $\langle \psi, f\rangle \geq (1/3) \cdot \|\psi\|_1$,
where $\|\psi\|_1 := \sum_{x \in\{-1,1\}^{RN}} |\psi(x)|$.
Finally, $\psi$ must equal 0 on inputs in $\pmone^{RN} \setminus \left(\pmone^{RN}\right)^{\leq N}$.

To simplify greatly, Bun et al.~\cite{BKT18} constructed their dual witness for $\left(\OR_R \circ \THR^k_N\right)^{\leq N}$ roughly as follows.
They took a dual witness $\Psi$ for 
the fact that $\adeg(\OR_R) \geq \Omega(R^{1/2})$
\cite{NS94, Spa08, BT15} 
and a dual witness 
$\phi$ for the fact that $\THR^k_N$ also
has large approximate degree, and they
combined $\Psi$ and $\phi$ in a certain
manner (introduced in prior works \cite{SZ09,She13, Lee09}) to get a dual witness for 
the composed function $\left(\OR_R \circ \THR^k_N\right)^{\leq N}$.
The technique used to combine $\Psi$ and $\phi$
is often called \emph{dual block composition},
and is denoted $\Psi \star \phi$.\footnote{To clarify, this entire
outline is a major simplification
of the actual dual witness construction in
\cite{BKT18}. The details provided in the outline of this introduction are chosen to highlight
the key technical issues that we must address
in this work. Amongst other
simplifications in this outline, the actual
dual witness from \cite{BKT18} is not
$\Psi \star \phi$, but rather a ``post-processed'' version of $\Psi \star \phi$, where the post-processing
step is used to ensure that the dual witness
evaluates to 0 on all inputs
of Hamming weight more than $N$.}
Dual block composition is defined as follows (below, each $x_i \in \{-1,1\}^N$):

\[
(\Psi \star \phi)(x_1, \dots, x_R) = 2^R \cdot 
\Psi(\sgn(\phi(x_1)), \dots, \sgn(\phi(x_R))) \cdot \prod_{i=1}^R |\phi(x_i)|/\|\phi\|_1.
\]

Here, $\sgn(r)$ equals
$-1$ if $r<0$ and equals $+1$ if $r>0$.\footnote{It is irrelevant how one
defines $\sgn(0)$
because if $\phi(x_i)=0$ for any $i$, the product
$\prod_{i=1}^R |\phi(x_i)|/\|\phi\|_1$ forces $\Psi \star \phi$ to 0. For this reason, the remainder of the discussion in this section implicitly assumes that $\phi(x_i) \neq 0$ for all $i \in \{1, \dots, R\}$.}
To show that $\Psi \star \phi$ is a dual witness
for the fact that the approximate degree of 
$\left(\OR_R \circ \THR_N^k\right)^{\leq N}$ is at least $d$,
it is necessary to show that 
 $\Psi \star \phi$ has
 pure high degree at least $d$, and that 
 $\Psi \star \phi$ is well-correlated with 
 $\left(\OR_R \circ \THR^k_N\right)^{\leq N}$.
 It is known that pure high degree
 increases multiplicatively under the $\star$ operation, and hence
 the pure high degree calculation
 for $\Psi \star \phi$ is straightforward. 
 In contrast, the correlation calculation is 
the key technical challenge and bottleneck
in the analysis of \cite{BKT18}.
Our key improvement over their work
is to modify the construction of the dual witness
in a manner that
allows for an improved correlation bound.

At a very high level,
what we do is replace the 
dual block composition $\Psi \star \phi$ from
the construction of \cite{BKT18}
with a \emph{variant} of dual block composition introduced by Sherstov \cite{SDPT}.
Sherstov specifically introduced this variant
to address the correlation issues that
arise when attempting to use dual block composition
to prove approximate degree lower bounds
for composed functions, and he used
it to prove direct sum and direct product
theorems for approximate degree.\footnote{Variants
of dual block composition related to the one introduced in \cite{SDPT} have played important
roles in other recent works on approximate 
degree lower bounds, e.g., \cite{largeerror, vanishingerror}.} However,
we have to modify even Sherstov's variant
of dual block composition in significant ways
to render it useful in our context.
We now attempt to give an informal sense of our modification and why it is necessary. 

For block-composed functions $g \circ h$,
the rough idea of any proof
attempting to show that $\langle \Psi \star \phi, g \circ h\rangle$ is large is to hope that
the following approximate equality holds:
\begin{equation}
    \label{approximateeq}
\langle \Psi \star \phi, g \circ h\rangle \approx
\langle \Psi, g\rangle.\end{equation}
If Equation \eqref{approximateeq} holds even approximately, then the correlation
analysis of $\Psi \star \phi$ is complete, since 
the assumption that $\Psi$ is a
dual witness for the high approximate degree of $g$
implies that the right hand side is large.

Equation \eqref{approximateeq} in fact holds 
with \emph{exact} equality if $\phi$ agrees in sign
with $h$ at all inputs, i.e., if $\langle \phi, h \rangle = \|\phi\|_1$ \cite{She13, Lee09}.
Unfortunately, the fact that $\phi$ is a dual
witness for the large approximate
degree of $h$ implies only a much weaker
lower bound on $\langle \phi, h \rangle$,
namely that
\begin{equation} \label{mylabel} \langle \phi, h \rangle \geq (1/3) \cdot \|\phi\|_1. \end{equation}
In general, Equation \eqref{mylabel}
is not enough to ensure that
Equation \eqref{approximateeq} holds even approximately. 

A rough intuition for why 
Equation \eqref{approximateeq} may fail to hold is the following. The definition of $\Psi \star \phi$
feeds $(\sgn(\phi(x_1)), \dots, \sgn(\phi(x_R)))$ 
into $\Psi$. One can think of 
$\sgn(\phi(x_i))$ as $\phi$'s ``prediction''
about $h(x_i)$, and the fact that 
$\langle \phi, h \rangle \geq (1/3) \cdot \|\phi\|_1$
means that for an $x_i$ chosen at random
from the probability distribution $|\phi|/\|\phi\|_1$,
this prediction is correct with probability at least $2/3$. Unfortunately, there are values of $x_i$
for which $\sgn(\phi(x_i)) \neq h(x_i)$, meaning that $\phi$'s  predictions can sometimes be wrong. In this case, in feeding $\sgn(\phi(x_i))$ into $\Psi$, dual block composition is ``feeding
an error'' into $\Psi$, and this can cause
$\Psi \star \phi$ to ``make more errors'' (i.e, output a value on an input that disagrees in sign with $g \circ h$ on that same input) than $\Psi$
itself. 

That is, there are two reasons $\Psi \star \phi$ may make an error: either $\Psi$ itself
may make an error (let us call this Source 1 for errors), and/or one or more copies of $\phi$ 
may make an error (let us call this Source 2 for errors).\footnote{There may be inputs $x=(x_1, \dots, x_n)$ to $\Psi \star \phi$ that could be classified as \emph{both} Source 1 and Source 2 errors. 
For purposes of this high-level introduction,
it is not important whether such inputs get classified as Source 1 or Source 2 errors for $\Psi \star \phi$.}

The first source of error
is already fully accounted for in the right hand side of Equation \eqref{approximateeq}. The second source of error is not, and this is the reason that Equation \eqref{approximateeq} may fail to hold even approximately.

Roughly speaking, while Equation \eqref{mylabel} guarantees
that $\sgn(\phi(x_i))$ is not ``an error''
for each $i$ with good probability (i.e., probability at least $2/3$), that still means that with very high probability, $\sgn(\phi(x_i))$ will be in error (i.e., not equal to $h(x_i)$) for \emph{a constant fraction} of blocks $i \in \{1, \dots, R\}$. Any one of these errors could
be enough to cause a Source 2 error. 

Fortunately for us, $g=\OR_R$ has low \emph{$(-1)$-certificate complexity}, meaning that on inputs $x$ in $\OR_R^{-1}(-1)$, to certify that indeed $x \in \OR_R^{-1}(-1)$, it is sufficient to identify
just one coordinate of $x$ that 
equals $-1$. 
This renders certain kinds of sign-errors made
by $\phi$ benign. 
Specifically, letting $S = \{x \colon \phi(x) < 0\}$ and $E^{-} = S \cap f^{-1}(1)$ denote
the false-negative errors made by $\phi$, the low $(-1)$-certificate complexity of $\OR_R$ means 
that it is okay if ``a constant
fraction of the negative values output by $\phi$
are in error''. That is, so long as
\begin{equation} \label{okeq} \left(\sum_{E^{-}} |\phi(x)|\right)/\left(\sum_{x \in S} |\phi(x)|\right) = 1-\Omega(1),\end{equation} 
the contribution of ``false negative errors made by $\phi$'' to actual Source 2 errors made by
$\Psi \star \phi$ is low.

However, the situation is starkly different
for ``false positive errors'' made by $\phi$; while $\OR_R$ has certificates of size 1
for inputs in $\OR_R^{-1}(-1)$, the certificate
complexity of the (unique) input in $\OR_R^{-1}(+1)$ is $n$. That is, letting $T = \{x \colon \phi(x) > 0\}$ and $E^{+} = T \cap f^{-1}(-1)$, for Equation \eqref{approximateeq} to hold even approximately for $g=\OR_R$,
it is essential that 
\begin{equation} \label{secondokeq} \left(\sum_{E^{+}} |\phi(x)|\right)/\left(\sum_{x \in T} |\phi(x)|\right) \ll 1/R.\end{equation} 

Accordingly, Bun et al.~\cite{BKT18} 
 obtain their lower bound for $k$-distinctness
by using a dual witness $\phi$ for
$h=\THR^k_N$
that satisfies Equation \eqref{secondokeq}. 
Using a dual with such few false positive errors causes \cite{BKT18} to lose an additive $1/(2k)$
term in the exponent of $N$ in their final degree bound,
relative to what they would obtain if Equation \eqref{mylabel} were sufficient to ensure that
Equation \eqref{approximateeq} approximately held.

As previously mentioned, Sherstov \cite{SDPT} introduced a variant of dual block composition 
intended to handle Source 2 errors 
that might have otherwise rendered Equation \eqref{approximateeq} false. 
Specifically, Sherstov proposed
multiplying $(\Psi \star \phi)(x)$ by a low-degree
polynomial $p_\eta(x)$ intended to ``kill''
any inputs $x$ that may contribute Source 2 errors (here, $\eta$ is a parameter, and
we will explain shortly
how the value of $\eta$ is ultimately chosen).
Specifically, $p_\eta$ ``counts'' 
the number of blocks $x_i$ of $x$
such that $\sgn(\phi(x_i)) \neq h(x_i)$, 
and $p_\eta$ is defined (through polynomial interpolation) to evaluate to 0 if this number is any integer between $1$ and $\eta$. 
This has the effect 
of eliminating all Source 2 errors 
made by $\Psi \star \phi$
on inputs $x$ for which
at most $\eta$ copies of $\phi$
make an error. That is, $p_\eta$ kills
all inputs $x$
in the set

\[
U_\eta :=\{ x = (x_1, \dots, x_R) \colon \sgn(\phi(x_i)) \neq h(x_i) \text{ for between } 1 \text{ and } \eta \text{ values of  } i\}.
\]
Note that multiplying $\Psi \star \phi$ by $p_\eta$
has the additional, unfortunate effect
of distorting the values that $\Psi \star \phi$ takes on other
inputs; bounding the effect of this distortion
is one challenge that Sherstov's
analysis (as well as our own analysis in this work) has to address. 

The intuition is that, so long as most Source 2 errors
made by $\Psi \star \phi$ are caused
by inputs in the set $U_\eta$, then 
multiplying $\Psi \star \phi$ by $p_\eta$ should eliminate the otherwise
devastating effects of most Source 2 errors.
So the remaining challenge
is to choose
a dual witness $\phi$ for $h$ guaranteeing that indeed most Source 2 errors are caused by inputs in $U_\eta$. More precisely, $\phi$
must be chosen to ensure that, with respect to the product distribution $\prod_{i=1}^R |\phi(x_i)|/\|\phi\|_1$,
it is very unlikely that more than 
$\eta$ copies of $\phi$ make an error on their input $x_i$.

To this end, it is implicit in Sherstov's
analysis that Equation \eqref{approximateeq} approximately holds with $\left(\Psi \star \phi\right) \cdot p_\eta$ in place of $\Psi \star \phi$ so long as 
\begin{equation} \label{secondokeq2} \left(\sum_{x \in E^{-} \cup E^{+}} |\phi(x)|\right)/\|\phi\|_1 \ll \eta/R.\end{equation} 
Notice that this is exactly Equation \eqref{secondokeq},
except that the right hand side has crucially increased by a factor of $\eta$ (also, Equation \eqref{secondokeq2} counts both false-positive and false-negative errors, as opposed to just false-positive errors, which is a key discrepancy that we address below).
The bigger that $\eta$ is set,
the less stringent is the requirement of Equation \eqref{secondokeq2}. However, it turns out that, in order to ensure that $\left(\Psi \star \phi\right) \cdot p_\eta$ 
has pure high degree close to that of $\Psi \star \phi$ itself,
$\eta$ must be set to a value that is noticeably smaller than the pure high degree
of $\Psi$. Ultimately, 
to obtain the strongest
possible results, $\eta$ gets
set to some constant 
$C<1$ times the pure high degree of $\Psi$.

In order to bring Sherstov's
ideas to bear on $k$-distinctness, we have to modify
his construction as follows.
The key issue (alluded to above) is that Sherstov's construction is not targeted at functions
$g \circ h$ where $g$ has low $(-1)$-certificate complexity,
and it is essential that we exploit
this low certificate complexity in the correlation analysis to improve
on the $k$-distinctness lower bound from \cite{BKT18}.
Essentially, we modify Sherstov's definition of $p_\eta$
to ``ignore'' all false negative errors (which as explained above are benign in our setting because $g=\OR_R$ has low $(-1)$-certificate complexity).
Rather we have $p_\eta$ only ``count'' the false positive errors
and kill any
inputs where this number is between $1$ and $\eta$. 

We are able to show that with this modification, it is sufficient to choose a dual witness $\phi$
for $\THR^k_N$
satisfying 

\begin{equation} \label{secondokeq3} \left(\sum_{E^{+}} |\phi(x)|\right)/\left(\sum_{x \in T} |\phi(x)|\right) \ll \eta/R.
\end{equation} 
We end up setting $\eta \approx O(\sqrt{R})$ for our lower bound, hence the denominator on the right hand side of this inequality represents a quadratic improvement compared to that on the right hand side of Equation \eqref{secondokeq}. This improvement ultimately
enables us to improve the lower bound
from $\tilde{\Omega}(N^{3/4-1/(2k)})$ to $\tilde{\Omega}(N^{3/4-1/(4k)})$.

The actual calculations required to establish 
the sufficiency of Equation \eqref{secondokeq3}
are quite involved, and
we provide a more detailed proof overview in
Section \ref{sec:detailedoutline} to help the reader make sense of them.

\subsection{The Upper Bound}
\label{s:proofoverviewupper}

Recall from Section \ref{lowerboundoverview}
that the approximate degree of $k$-distinctness 
is (essentially) equivalent to $\adeg(f_{R N}^{\leq N})$ for
$f=\OR_R \circ \THR^k_N$.
Similarly, the approximate degree of the Surjectivity
function is (essentially) equivalent
to $\adeg(f_{R N}^{\leq N})$ for
$f=\AND_R \circ \OR_N$.
Sherstov proved an upper bound of $O(R^{1/4} \cdot N^{1/2})$ for this latter quantity.

Up to polylogarithmic factors, in Theorem \ref{thm:informalupper} we achieve an identical upper bound for $k$-distinctness,
for any $k \leq \text{polylog}(N)$.
To do so, we make the following
easy observations.
First, in order to apply Sherstov's
construction to a function $f=g \circ h$,
it is enough that $g$ have approximate 
degree $O(\sqrt{R})$,\footnote{More precisely,
it should be possible to approximate $g$
by a linear combination of monotone conjunctions,
where the $\ell_1$-norm of the coefficients
of the linear combination is $2^{\tilde{O}(\sqrt{R})}$.
It is not hard to show, by Parseval's identity, that this is guaranteed if $g$ has approximate degree $\tilde{O}(\sqrt{R})$.}
and that $h$ be exactly computed
as a linear combination of conjunctions,
where the coefficients in the linear combination
have $\ell_1$-norm at most quasipolynomially large
in $N$.
Second, we observe that for $k \leq \text{polylog}(N)$, $\THR^k_N$ is exactly
computed by such a linear combination
of conjunctions. Together,
these observations are enough to apply Sherstov's construction for Surjectivity to obtain the
approximate degree upper bound of Theorem \ref{thm:informalupper} for $k$-distinctness.

\section{Preliminaries}
\label{sec:prelims}

\paragraph{Notation.} Let $N, n$ and $m$ be positive integers, $N\leq n$. For $z\in\pmone^n$, let $|z|$ represent the \emph{Hamming weight} of $z$, i.e.,~the number of $-1$'s in $z$. 
Define $(\pmone^{n})^{\leq N} := \bra{x \in \pmone^n : |x| \leq N}$.
For any function $f : \pmone^n \to \R$, denote by $f^{\leq N}$ the partial function that is defined on $(\pmone^n)^{\leq N}$ and agrees with $f$ on all such inputs.
Define $\sgn:\R \to \pmone$ by $\sgn(x) = 1$ for all non-negative $x$, and $-1$ otherwise. All logarithms in this paper are base 2 unless otherwise specified. Let $1^n$ (respectively, $(-1)^n$) denote the $n$-bit string $(1, 1, \dots, 1)$ (respectively, $(-1, -1, \dots, -1)$). For strings $a \in \pmone^m$ and $b \in \pmone^n$, we denote by $a, b$ the $(m+n)$-bit string formed by the concatenation of $a$ and $b$.
We use the notation $[n]$ to denote the set $\bra{1, 2, \dots, n}$.

For any function $f : \pmone^n \to \R$, define $\|f\|_1 = \sum_{x \in \pmone^n}|f(x)|$. For an event $E$, the corresponding indicator function is 
\begin{equation}\label{eqn: indicator}
    I[E] = 
    \begin{cases*}
      1 & if $E$ holds, \\
      0    & otherwise.
    \end{cases*}
\end{equation}

For any function $\psi:\pmone^m\to\R$ such that $\norm{\psi}_1=1$, let $\mu_\psi$ be the distribution on $\pmone^{m}$, defined by 
\begin{equation}\label{eqn: mudef}
\mu_\psi(x) = |\psi(x)|.
\end{equation}

\begin{definition}\label{def: mu_z}
For any integer $n > 0$, any function $\psi:\pmone^m\to\R$ such that $\norm{\psi}_1=1$, and any $w \in \pmone$, let $\mu_{w}$ be the probability distribution $\mu_{\psi}$ conditioned on the event that $\sgn(\psi(x)) = w$. For any $z \in \pmone^n$, let $\mu_z$ denote the probability distribution $(\mu_\psi)^{\otimes n}$ conditioned on the event that $\sgn(\psi(x_i))=z_i$ for all $i\in[n]$.
\end{definition}

We omit the dependence of $\mu_z$ on $\psi$ since $\psi$ will typically be clear from context.  Note that $\mu_z$ as defined above is a product distribution given by
\begin{equation}\label{eqn: mu_z_prod_dist}
    \mu_{z}(x_1,\dots,x_n) = \prod_{i=1}^n \mu_{z_i}(x_i).
\end{equation}

\begin{definition}\label{def: prod_dist_Pi}
 For $\eta_i\in[0,1]$, let
$\Pi(\eta_1,\dots,\eta_n)$ be the product distribution on $\pmone^n$ where the $i$th bit of the string equals $-1$ with probability $\eta_i$, and $1$ with probability $1 - \eta_i$. \end{definition}

\begin{lemma}\label{lem: expectation_multilinear}
Let $n$ be any positive integer, $p: \pmone^{n}\to \R$ be a multilinear polynomial, and $\eta_1, \dots, \eta_n\in[0,1]$. For $x=(x_1,\dots,x_n)$ drawn from the product distribution $\Pi(\eta_1,\dots,\eta_n)$ defined in Definition~\ref{def: prod_dist_Pi}, we have
\begin{equation}
    \E_{\Pi(\eta_1,\dots,\eta_n)}[p(x_1,\dots,x_n)] = p(1-2\eta_1,\dots,1-2\eta_n).
\end{equation}
\end{lemma}

Any function $f : \pmone^n \to \R$ has a unique multilinear representation $f = \sum_{S \subseteq [n]} \hat{f}(S) \chi_S$, where for any $S \subseteq [n]$, the function $\chi_S : \pmone^n \to \pmone$ is defined by $\chi_S(x) = \prod_{i \in S}x_i$. Hence, $\|\hat{f}\|_1 = \sum_{S \subseteq [n]} |\hat{f}(S)|$. It follows that for any function $\phi:\pmone^n\to\R$, there exists a unique multilinear polynomial $\tilde\phi:\R^n\to\R$ such that $\tilde \phi(x) = \phi(x)$ for all $x \in \pmone^n$. 

\subsection{Functions of Interest}

Define the function $\OR_N:\pmone^{N}\to\pmone$ to equal $1$ if $x = 1^N$, and $-1$ otherwise.
Define the \emph{Threshold} function $\THR_N^k:\pmone^{N}\to\pmone$ to equal $1$ for inputs of Hamming weight less than $k$, and $-1$ otherwise.

\begin{definition}[$k$-distinctness]
For integers $k, N, R$ with $k\leq N$, define the function $\DIST^k_{N,R}: [R]^N\to\pmone$ by $\DIST^k_{N,R}(s_1,\dots,s_N)=-1$ iff there exists an $r\in[R]$ and distinct indices $i_1,\dots,i_k$ such that $s_{i_1}=\dots=s_{i_k}=r$. When necessary, the domain of the function can be viewed as $\pmone^{N\log R}$.
\end{definition}

Given any functions $f_n : \bra{-1, 1}^n \rightarrow \bra{-1, 1}$ and $g_m : \bra{-1, 1}^m \rightarrow \bra{-1, 1}$, we define the function $f_n \circ g_m : \bra{-1, 1}^{mn} \rightarrow \bra{-1, 1}$ as $f_n \circ g_m (x_{11}, \dots, x_{1m}, x_{21}, \dots, x_{2m}, \dots, x_{n1}, \dots, x_{nm}) = f_n(g_m(x_1), g_m(x_2), \dots, g_m(x_n)), x_i\in\pmone^m$ for all $i\in[n]$. We drop subscripts when the arities of the constituent functions are clear.

\subsection{Notions of Approximation}
\begin{definition}[Approximate degree]
For any function $f : \pmone^n \to \R$, any integer $N\leq n$, and any $\epsilon \in [0, 1]$, define the $\epsilon$-approximate degree of $f^{\leq N}$ to be
\[
\adeg_{\epsilon}(f^{\leq N}) 
 = \min_{\substack{   p : |{p(x) - f(x)}| \leq \epsilon \\\forall x \in \pmone^n ,|x|\leq N}} \deg(p).
\]
When the subscript is dropped, $\epsilon$ is assumed to equal $1/3$. When the superscript is dropped in $f^{\leq N}$, then $N$ is assumed to equal $n$.\footnote{Note that this definition places no constraints on an approximating polynomial on inputs outside the promise domain. In other contexts, an approximating polynomial may be required to be bounded outside the promise domain. }
\end{definition}

\begin{definition}
For any finite subset $X \subseteq \R^n$, any function $f : X \to \R$, and any integer $d \geq 0$, define
\[
E(f, d) := \min_{p : \deg(p) \leq d}\left\{\max_{x \in X}|f(x) - p(x)|\right\}.
\]
\end{definition}

\begin{definition}[Correlation]\label{def: correlation}
Consider any function $f : \pmone^n \to \R$ and $\psi: \pmone^n\to\R$. Define the \emph{correlation} between $f$ and $\psi$ to be 
\[
\langle f,\psi \rangle = \sum_{x\in\pmone^n}f(x)\psi(x).
\]
\end{definition}

\begin{definition}[Pure high degree]
For $\phi:\pmone^n\to\R$, we say that the \emph{pure high degree} of $\phi$, which we denote by $\phd(\phi)$, is $d$ if $d \geq 0$ is the largest 
integer for which $\langle \phi, p \rangle =0$ for any polynomial $p:\pmone^n\to\R$ of degree strictly less than $d$. 
\end{definition}

For any Boolean function $f:\pmone^m\to\pmone$ and function $\psi: \pmone^m \to \R, \norm{\psi}_1=1$, let
\begin{equation}\label{eqn: epsilon}
   \epsilon_{f, \psi}^+ := \Pr_{\mu_\psi}[f(x)\psi(x)<0|\psi(x)>0], \qquad \epsilon_{f, \psi}^- := \Pr_{\mu_\psi}[f(x)\psi(x)<0|\psi(x)<0].
\end{equation}

Define $\epsilon_{f, \psi}=\epsilon_{f, \psi}^++\epsilon_{f, \psi}^-$.

\begin{definition}\label{def: delta}
For any functions $f:\pmone^n\to\pmone$ and $\psi:\pmone^n\to\R$, let  \begin{align*}
E^+(f,\psi)& := \{x\in\pmone^n: f(x)\psi(x)<0, \psi(x)>0 \},\\
E^-(f,\psi) & := \{x\in \pmone^n: f(x)\psi(x)<0, \psi(x)<0 \}.
\end{align*}
We define the \emph{false positive error} between $f$ and $\psi$ to be
\[
\delta^+_{f,\psi}:=\sum_{x\in E^+(f,\psi)}|\psi(x)|
\]
and \emph{false negative error} to be \[
\delta^-_{f,\psi}:=\sum_{x\in E^-(f,\psi)}|\psi(x)|.
\]
\end{definition}

We observe the following simple connection between $\delta^+_{f, \psi}~(\delta^-_{f, \psi})$ and $\epsilon^+_{f, \psi}~(\epsilon^-_{f, \psi})$.
\begin{claim}\label{clm: epsilondeltarelation}
For any Boolean function $f:\pmone^m\to\pmone$ and any function $\psi:\pmone^m\to\R$ with $\norm{\psi}_1 = 1, \phd(\psi)\geq 1$, 
\begin{equation}
    \epsilon^+_{f, \psi} = 2\delta^+_{f, \psi},\quad \epsilon^-_{f, \psi} = 2\delta^-_{f,\psi}.
\end{equation}
\end{claim}
\begin{proof}
\begin{align*}
    \delta^+_{f, \psi} &= \sum_{x\in E^+(f,\psi)}|\psi(x)| \tag*{by Definition~\ref{def: delta}}\\
    &= \Pr_{x\sim{\mu_\psi}}[x\in E^+(f,\psi)]\tag*{by Equation~\eqref{eqn: mudef}}\\
    &=\Pr_{x\sim{\mu_\psi}}[f(x)\psi(x)<0 \land \psi(x)>0] \tag*{by Definition~\ref{def: delta}}\\
    &= \Pr_{x \sim \mu_\psi}[\psi(x)>0]\cdot\Pr_{x \sim \mu_\psi}[f(x)\psi(x)<0|\psi(x)>0]\\
    &= \frac{\epsilon^+_{f, \psi}}{2} \tag*{since $\langle \psi, 1 \rangle = 0$ and $\sum_x|\psi(x)|=1$ implies $\Pr_{\mu_\psi}[\psi(x)>0]=1/2$}.
\end{align*}
The equality $\epsilon^-_{f, \psi} = 2\delta_{f,\psi}^-$ can be proved similarly.
\end{proof}

By linear programming duality, we have the following standard equivalence between lower bounds on approximate degree and existence of ``dual polynomials". See, for example,~\cite{BKT18Arxiv}.
\begin{lemma}\label{lem: dualitypromise}
Let $f : \pmone^n \to \pmone$ be any function.  For any integer $0 \leq j \leq n$, we have 
 $\adeg_{\epsilon}(f^{\leq j}) \geq d$ if and only if there exists a ``dual polynomial'' $\phi : \pmone^n \to \R$ satisfying the following properties.
\begin{itemize}
    \item $\sum_{x\in\pmone^n}|\phi(x)|=1$.
    \item $\phd(\phi)>d$.
    \item $\langle f,\phi \rangle>\epsilon$.  
    \item $\phi(x) = 0$ for all $|x|>j$.
\end{itemize}
We say that $\phi$ is a dual polynomial witnessing the fact that $\adeg_{\epsilon}(f^{\leq j}) > d$. For brevity, when $\epsilon$ and $d$ are clear from context, we say that $\phi$ is a dual polynomial for $f^{\leq j}$.
\end{lemma}

\v{S}palek~\cite{Spa08} exhibited an explicit dual witness for $\OR$ (existence of a dual witness for $\OR$ was already implicit from the work of Nisan and Szegedy~\cite{NS94}).
\begin{claim}[Implicit in \cite{NS94}]\label{clm: canonical_dual_OR}
There exists a constant $c \in (0, 1]$ such that for any integer $n \geq 0$, there exists a function $\theta : \pmone^n \to \R$ satisfying
\begin{itemize}
    \item $\|\theta\|_1 = 1$,
    \item $\phd(\theta) \geq c\sqrt{n}$,
    \item $\langle \theta, \OR_n \rangle \geq 3/5$.
\end{itemize}
\end{claim}

We also require the following error reduction theorem for approximate degree.
\begin{lemma}[{\cite{BNRdW07}}]\label{lem: adegamp}
Let $f : \pmone^n \to \pmone$ be any (possibly partial) Boolean function and let $0 < \epsilon < 1$. Then,
\[
\adeg_\epsilon(f) = \adeg(f) \cdot O(\log(1/\epsilon)).\footnote{The statement in~\cite{BNRdW07} only deals with total functions. It can be seen that the proof works for partial functions too.}
\]
\end{lemma}

\subsection{Dual Polynomials and Dual Block Composition}

Bun et al.~\cite{BKT18} exhibited a dual witness for the approximate degree of the $k$-threshold function. Their dual witness additionally satisfies a decay condition, meaning that it places very little mass on inputs of large Hamming weight.
The following claim, which gives a preliminary construction towards their dual witness for $\THR^k_N$, is a mild modification of~\cite[Proposition 54]{BKT18Arxiv}. 

\begin{claim}[Modification of {\cite[Proposition 54]{BKT18Arxiv}}]\label{clm: psiconstruction}
Let $k, T, N\in\N$ with $2 \leq k\leq T$. There exist constants $c_1, c_2\in (0,1]$ and a function $\omega_T:[T]\cup \{0\}\to\R$ such that all of the following hold.
\begin{align}
    & \sum_{\omega_T(t)>0,t\geq k}|\omega_T(t)| \leq \frac{1}{48\cdot 4^k  \sqrt{N} \log N}.\\
    \label{eqn: delta^+ bound}
    & \sum_{\omega_T(t)<0,t<k}|\omega_T(t)| \leq \left(\frac{1}{2}-\frac{2}{4^k} \right).\\
    & \norm{\omega_T}_1 := \sum_{t=0}^T|\omega_T(t)|=1.\\
    & \text{For all polynomials $q:\R\rightarrow \R$},\nonumber\\
    & \deg (q) < c_1\sqrt{4^{-k}k^{-1}TN^{-1/(2k)}\log^{-1}N} \implies \sum_{t=0}^T\omega_T(t)q(t)=0.\\
    & \text{For all $t \in [T]$},
    |\omega_T(t)| \leq \frac{\sigma\exp(-\beta t)}{t^2}
    \qquad \text{for~} \sigma= (2k)^k, \quad \beta= c_2/\sqrt{4^kkTN^{1/(2k)} \log N}.\label{eqn: omega_weak_decay}
\end{align}
\end{claim}
Although the proof follows along the same lines as that of~\cite{BKT18Arxiv}, we provide a proof in the appendix for completeness.

The next claim yields a dual polynomial for $\THR^k_N$, and we omit its proof.

\begin{claim}[Modification of {\cite[Proposition 55]{BKT18Arxiv}}]\label{clm: final_psi_construction}
Let $k, T, N\in \N$ with $2 \leq k \le T \le N$, and let $\omega_T$ be as constructed in Claim~\ref{clm: psiconstruction}, with constants $c_1, c_2$. Define $\psi_T : \pmone^N \to \R$ by $\psi_T(x) = \omega_T(|x|) / \binom{N}{ |x|}$ for $x \in (\pmone^N)^{\leq T}$ and $\psi_T(x) = 0$ otherwise. Then 
\begin{align}
&\delta^+_{ \THR^k_N, \psi_T} \le \frac{1}{48\cdot 4^k  \sqrt{N} \log N}  \label{eqn:thr-pos-corr}\\
&\delta^-_{\THR^k_N, \psi_T} \le \frac{1}{2} - \frac{2}{4^k}  \label{eqn:thr-neg-corr} \\
&\|\psi_T\|_1 = 1   \label{eqn:thr-norm} \\
&\text{For any polynomial } p \colon \pmone^N \to \R\text{, } \nonumber\\
&\deg (p) <  c_1 \sqrt{4^{-k} k^{-1} T N^{-1/(2k)} \log^{-1} N} \implies \langle \psi_T, p \rangle = 0 \label{eqn:thr-phd} \\
&\text{For all~}t \in [n], \qquad \sum_{|x| = t} |\psi_T(x)|  \le \frac{(2k)^k \exp\left(-c_2 t / \sqrt{4^kk T N^{1/(2k)} \log N}\right)}{t^2}.\label{eqn:thr-decay}
\end{align}
\end{claim}

Towards proving approximate degree lower bounds for composed functions, one might hope to combine dual polynomials of the constituent functions in some way to obtain a dual polynomial for the composed function. A series of works~\cite{SZ09,Lee09,She13} introduced the notion of ``dual block composition'', which is a powerful method of combining dual witnesses.
\begin{definition}[Dual block composition]\label{defn: dualblockcomposition}
Let $\theta:\pmone^n \to \R, \phi : \pmone^m \to \R$ be any functions satisfying $\|\theta\|_1 = \|\phi\|_1 = 1$ and $\phd(\phi) \geq 1$.
Let $x=(x_1, \dots, x_n)$ where each $x_i \in \pmone^m$.
Define the \emph{dual block composition} of $\theta$ and $\phi$, denoted $\theta \star \phi$, to be
\[
\theta \star \phi(x) =2^n \theta(\sgn(\phi(x_1)),\dots, \sgn(\phi(x_n)))\prod_{i=1}^n |\phi(x_i)|.
\]
\end{definition}
Sherstov~\cite{She13} showed that dual block composition preserves $\ell_1$-norm and that pure high degree is multiplicative (also see~\cite{Lee09}). Bun and Thaler~\cite{BT17} observed that dual block composition is associative.
\begin{lemma}\label{lem: dbproperties}
Let $\phi:\pmone^{m_\phi}\to \R, \theta:\pmone^{m_\theta}\to\R$ be any functions. Then,\\
\textbf{Preservation of $\ell_1$-norm}: If $\norm{\theta}_1=1, \norm{\phi}_1=1$ and $\langle \phi,1 \rangle=0$, then 
\begin{equation}
    \norm{\theta\star\phi}_1 = 1.
\end{equation}
\textbf{Multiplicativity of pure high degree}: 
\begin{equation}\label{eqn: phdmult}
\phd(\theta)>D, \phd(\phi)>d \implies \phd(\theta\star\phi)>Dd.
\end{equation}
\textbf{Associativity}: 
For every $\psi:\pmone^{m_\psi}\to\R$, we have
\begin{equation}\label{eqn: dbcassociative}
    (\phi\star\theta)\star\psi = \phi\star(\theta\star\psi).
\end{equation}
\end{lemma}

It was shown in~\cite{BKT18Arxiv} that for any dual polynomial $\Phi$, and $\psi_T$ as constructed in Claim~\ref{clm: final_psi_construction}, the dual block composed function $\Phi \star \psi_T$ satisfies a ``strong dual decay'' condition.\footnote{They in fact showed that $\Psi \star \psi$ satisfies this strong decay condition for \emph{any} $\psi$ satisfying a corresponding ``weak decay'' condition. However for this paper, we only require this statement for $\psi = \psi_T$ as constructed in Claim~\ref{clm: final_psi_construction}.} 

\begin{claim}[{\cite[Proposition 31]{BKT18Arxiv}}]\label{clm: strongdualdecay}
Let $R$ be sufficiently large and $k\leq T \leq R$ be any positive integer. Fix $\sigma = (2k)^k$ and let $N = \ceil{20\sqrt{\sigma}R}$. Let $\Phi:\pmone^R\rightarrow \R$ be any function with $\norm{\Phi}_1 = 1$ and $\psi_T:\pmone^N \rightarrow \R$ as defined in Claim~\ref{clm: final_psi_construction}. Then
\begin{equation}
  \sum_{x\notin (\pmone^{RN})^{\leq N}}|(\Phi\star\psi_T)(x)|\leq (2NR)^{-2\Delta}
\end{equation}
for some $\Delta \geq \frac{\beta\sqrt{\sigma}R}{4\ln^2{R}}$ for $\beta= c_2/\sqrt{4^kkTN^{1/(2k)} \log N}$.
\end{claim}

We now define a simple 
but important function $\phi$ that we use in
our construction of a dual witness for 
$\DIST^k_{N,R}$. This function
was first used in the context 
of dual block composition
by Bun and Thaler \cite{BT15}. 
\begin{claim}[\cite{BT15}]\label{claim: halfhalf}
Define $\phi:\pmone^n\to\R$ as
  \begin{equation}
    \phi(x) =
    \begin{cases*}
      -1/2 & if $x = -1^n$ \\
      1/2  & if $x= 1^n$ \\
      0    & otherwise.
    \end{cases*}
  \end{equation}
  Then, $\phd(\phi) = 1$.
\end{claim}

Bun et al.~\cite{BKT18Arxiv}, slightly extending a result in \cite{BT15}, showed that on dual block composing $\phi$ and $\psi$, where $\phi$ is defined as in Claim~\ref{claim: halfhalf}, the correlation of the dual block composed witness $\phi \star \psi$ with $\OR_M \circ f$ amplifies the correlation of $f$ with $\psi$ as follows.

\begin{lemma}[{\cite[Proposition 56]{BKT18Arxiv}}]\label{lem: halfhalfamplification}
Let $f: \pmone^n\to \pmone$ and $\psi: \pmone^n \to \R$ be any functions with $\|\psi\|_1 = 1$.
For every $M\in\N$ and $\phi: \pmone^M \to \R$ as defined in Claim~\ref{claim: halfhalf}, we have
\begin{equation}\label{eqn: delta_plus_amplification}
    \delta^+_{ \OR_M \circ f, \phi\star\psi}\leq M \delta^+_{f,\psi},
\end{equation}
\begin{equation}\label{eqn: delta_minus_amplification}
    \delta^-_{ \OR_M \circ f, \phi\star\psi}\leq \frac{1}{2}(2\delta^-_{f,\psi})^M.
\end{equation}
\end{lemma}

\subsection{Some Polynomials}

In this section we list out a few polynomials that we require, along with their properties.

\begin{lemma}[{\cite[Lemma 3.1]{SDPT}}]\label{lem: sherstovpk}

For any $\tau_1,\dots,\tau_n \in [0,1)$, define $\nu = \Pi(\tau_1,\dots,\tau_n)$ and $\tau=\max\{\tau_1,\dots,\tau_n \}$. For any $\eta=0,1,\dots, n-1$, let $p_{\eta}:[-1,1]^n\to\R$ be the unique degree-$\eta$ multilinear polynomial that satisfies
\begin{equation}
    p_{\eta}(z)=(-1)^{\eta}\prod_{i=1}^{\eta}(|z|-i), \forall z\in\pmone^n.
\end{equation}
Then,
\begin{align}
&p_{\eta}(1^n) = \eta!, \label{eqn: p_eta_all_false} \\
& \norm{\hat p_\eta}_1\leq \eta!\binom{n+\eta}{\eta}, \label{eqn: p_hat_bound}\\
    &\E_\nu[|p_{\eta}(z)|] \leq p_{\eta}(1^n)\nu(1^n)\left( 1+A\right),\quad
    \text{where}~A:=\binom{n}{\eta+1}\frac{\tau^{\eta+1}}{(1-\tau)^n}.\label{eqn: p_eta_bound}
\end{align}
Furthermore, $p_\eta(z)\geq 0$ for all $z\in\pmone^n$ provided that $\eta$ is even.
\end{lemma}

It is easy to show that for any multilinear polynomial $p : \R^n \to \R$, we have $\max_{y \in [-1, 1]^n}|p(y)| \leq \|\hat{p}\|_1$.  When applied to the function in the previous lemma, we obtain
\begin{claim} \label{claim: pkupperbound}
For $p_\eta$ defined as in Lemma~\ref{lem: sherstovpk}, $\max_{y\in [-1,1]^n} |p_{\eta}(y)|  \leq \eta!\binom{n+\eta}{\eta}.$
\end{claim}

Finally, we require a lemma, implicit in a result of Razborov and Sherstov~\cite{RS10} (also see~\cite[Proposition 21]{BT17} for a formulation similar to the one we require), that helps us convert a dual polynomial with little mass on large Hamming weight inputs to a dual polynomial with no mass on large Hamming weight inputs without affecting the pure high degree by much.
\begin{lemma}[Implicit in~\cite{RS10}]\label{lem: RS}
Let $N\geq R$ be positive integers, $\Delta\in \R^+$, and $\theta:\pmone^{RN}\to\R$ be any polynomial such that 
\[
    \sum_{x\notin (\pmone^{RN})^{\leq N}} |\theta(x)|\leq (2NR)^{-\Delta}.
\]
For any positive integer $D<\Delta$, there exists a function $\nu:\pmone^{RN}\to\R$ such that 
\begin{itemize}
    \item $\phd(\nu)>D$
    \item $\norm{\nu}_1\leq 1/10$
    \item $|x|>N\Rightarrow \nu(x)=\theta(x)$.
\end{itemize}
\end{lemma}

\begin{definition}\label{defn: cheb}
For any integer $d \geq 0$, let $T_d : \R \to \R$ denote the degree-$d$ Chebyshev polynomial, defined recursively as follows.
\begin{align*}
T_0(x) & = 1\\
T_1(x) & = x\\
T_d(x) & = 2xT_{d-1}(x) - T_{d-2}(x).
\end{align*}
\end{definition}

We now observe a simple well-known fact about Chebyshev polynomials whose proof we include for completeness.

\begin{claim}\label{claim: coeffs}
For any $d \geq 0$, consider the $d$'th Chebyshev polynomial $T_d : \R \to \R$ as defined in Definition~\ref{defn: cheb}, and write its expansion $T_d(x) = \sum_{i = 0}^d a_ix^i$. Then, 
\begin{equation}\label{eqn: chebcoeffs}
\sum_{i = 0}^d |a_i| \leq 3^d.
\end{equation}
\end{claim}

\begin{proof}
We prove this by induction.

Let $S_d$ denote $\sum_{i = 0}^d |a_i|$ where $a_i$'s are the coefficients in the expansion $T_d(x) = \sum_{i = 0}^d a_ix^i$.
By Definition~\ref{defn: cheb}, the hypothesis is satisfied for $d = 0, 1$. Next suppose the hypothesis is true for all $d \leq k$ for some $k \geq 1$. By the recursive definition in Definition~\ref{defn: cheb}, we have $S_{k+1} \leq 2S_k + S_{k - 1} \leq 2 \cdot 3^{k} + 3^{k-1} = 3^{k - 1}(6 + 1) < 3^{k+1}$.
\end{proof}

We also require the following well-known properties of Chebyshev polynomials.
\begin{fact}\label{fact: chebbds}
For any integer $d \geq 0$,
\begin{align}
|T_d(x)| & \leq 1 & |x| \leq 1\label{eqn: abs1}\\
T_d(1 + \epsilon) & \geq 1 + d^2\epsilon & \epsilon \geq 0.\label{eqn: largeoutside1}
\end{align}
\end{fact}

\begin{definition}\label{def: conjunction}
For any positive integer $n$, any polynomial $p:\pmone^n \to \{0,1\}$ that is of the form 
\begin{equation}
    \left(\prod_{i \in A}\frac{1+x_i}{2}\right)  \left(\prod_{j \in B}\frac{1-x_j}{2}\right)
\end{equation}
for some sets $A, B \subseteq [n]$, is called a conjunction.
\end{definition}

It can be observed that the product of conjunctions is a conjunction. 

\begin{claim}[{\cite[Corollary 4.7]{algopoly}}]\label{claim: conjbddHW}

Let $n \leq N$ be any positive integers, and $A, B$ be any subsets of $[N]$. Define $f : \left(\pmone^N\right)^{\leq n} \to \zone$\footnote{The version in~\cite{algopoly} deals with functions whose domain is $\left(\zone^N\right)^{\leq n}$. The statement there can easily be seen to imply the statement in this paper.} by
\[
f(x) = \left(\prod_{i \in A}\frac{1+x_i}{2}\right)  \left(\prod_{j \in B}\frac{1-x_j}{2}\right).
\]
Then, for any integer $d \geq 0$, we have
\[
E(f, d) \leq \exp\left(-\frac{cd^2}{n}\right)
\]
for some absolute constant $c$.
\end{claim}

\begin{definition}\label{def: norm_conjunction}
Consider any positive integer $n$ and any function $f:\pmone^n\to\R$. Define the \emph{conjunction norm} of $f$, which we denote by $\rho(f)$, to be 
\[
\min\bra{\sum_{A\subseteq{[n]}}\sum_{B\subseteq{[n]}} |C_{A,B}| : f(x) = \sum_{A\subseteq{[n]}}\sum_{B\subseteq{[n]}} C_{A,B} \left(\prod_{i \in A}\frac{1+x_i}{2}\right)  \left(\prod_{j \in B}\frac{1-x_j}{2}\right), \quad C_{A,B}\in \R}.
\]

\end{definition}

We now state some simple observations about the conjunction norm which we do not prove here. See, for example,~\cite[Proposition 2.4]{algopoly}.
\begin{fact}

Let $m, n$ be positive integers, $f,g:\pmone^n \to \R$ be any functions, and $p:\R\to\R$ be any degree-$m$ polynomial of the form $p(x)=\sum_{i=0}^m a_ix^i, a_i\in \R$. Then $\rho$ is well defined and satisfies
\begin{align}
    \rho(a\cdot f) & = |a|\rho(f),\quad \text{for any $a\in\R$},\label{eqn: rho_constant_mult}\\
    \rho(f+g) &\leq \rho(f) + \rho(g),\label{eqn: rho_add}\\
    \rho(f\cdot g) &\leq \rho(f) \cdot \rho(g), \label{eqn: rho_mult}\\
    \rho(p \circ g) &\leq (\max\{1, \rho(g)\})^m\cdot \sum_{i=0}^m|a_i|.\label{eqn: rho_circ}
\end{align}
\end{fact}


\section{Outline of Proof of Main Theorem}
\label{sec:detailedoutline}
Our main theorem is as follows.

\begin{theorem}\label{thm: main}
For $R \in \mathbb{N}$ sufficiently large, $2 \leq k\leq \frac{\log R}{4}$, and some $N=\Theta(k^{k/2}R)$, \begin{equation}
\adeg(\DIST^k_{N,R+N}) = \Omega\left(\frac{1}{4^k k^2}\cdot\frac{1}{\log^{7/2}R}\cdot R^{\frac{3}{4}-\frac{1}{4k}}\right).
\end{equation}
\end{theorem}

Ambainis~\cite{ambainissmallrange} showed that the approximate degree\footnote{There are several
different conventions used in the literature when
 defining
the domain of functions such as $k$-distinctness.
The convention used by Ambainis~\cite{ambainissmallrange} considers the input to be specified 
by $N \cdot R$ variables $y_{1, 1}, \dots, y_{N, R}$, where $y_{i, j}=-1$ if and only if the $i$th list
item in the input equals range element $j$
(i.e., it is promised that for each $i$, $y_{i, j}=-1$
for exactly one $j$).
We use the convention that
the input is specified by $N \lceil \log_2 R\rceil$
bits. It is well known (and not hard to show) that 
conversion between the two conventions affects approximate degree by at most a factor of $\lceil \log_2 R \rceil$.} of functions that are symmetric (both with respect to range elements and with respect to domain elements) is the same for all range sizes greater than or equal to $N$. As a corollary, we obtain the following.

\begin{corollary}\label{cor: ambrangereduction}
For $R \in \mathbb{N}$ sufficiently large, $2 \leq k\leq \frac{\log R}{4}$, and some $N=\Theta(k^{k/2}R)$, \begin{equation}
\adeg(\DIST^k_{N,N}) = \Omega\left(\frac{1}{4^k k^2}\cdot\frac{1}{\log^{7/2}R}\cdot R^{\frac{3}{4}-\frac{1}{4k}}\right).
\end{equation}
\end{corollary}

We require the following relation between approximate degree of $k$-distinctness and a related Boolean function; this relationship follows from \cite[Proposition 21 and Corollary 26]{BKT18Arxiv}.

\begin{claim}[\cite{BKT18Arxiv}]\label{clm: dist_connection}
Let $N, R\in\N$ and $2\leq k \leq N$ be any integer. Then for any $\epsilon>0$,
\begin{equation}
    \adeg_\epsilon(\DIST^k_{N,R+N}) = \Omega\left(\frac{1}{\log R}\cdot\adeg_\epsilon(\OR_R\circ\THR^k_N)^{\leq N}\right).
\end{equation}
\end{claim}

To prove Theorem~\ref{thm: main}, Claim~\ref{clm: dist_connection} implies that it suffices to prove a lower bound on $\adeg(\OR_{R} \circ \THR^k_N)^{\leq N}$. 

\begin{theorem}\label{thm: mainsummary}
For $R \in \mathbb{N}$ sufficiently large, $2 \leq k\leq \frac{\log R}{4}$, and some $N=\Theta(k^{k/2}R)$,
\begin{equation}
   \adeg((\OR_{R} \circ \THR^k_N)^{\leq N}) = \Omega\left(\frac{1}{4^k k^2}\cdot\frac{1}{\log^{5/2}R}\cdot R^{\frac{3}{4}-\frac{1}{4k}}\right). 
\end{equation}
\end{theorem}

 Note that the theorems above continue to yield non-trivial lower bounds for some values of $k=\omega(1)$. However for ease of exposition, we assume throughout this section that $k \geq 2$ is an arbitrary but fixed constant.

\medskip \noindent \textbf{Outline of the Proof of Theorem~\ref{thm: mainsummary}.} Towards proving Theorem~\ref{thm: mainsummary}, we construct a dual witness $\Gamma$ satisfying the following four conditions. 

\begin{itemize}
    \item \textbf{Normalization:} $\norm{\Gamma}_1=1$,
    \item \textbf{Pure high degree:} There exists a $D = \tilde{\Omega}\left(R^{\frac{3}{4} - \frac{1}{4k}}\right)$ such that for every polynomial $p:\pmone^{RN}\to\R$ of degree less than $D$, we have $\langle p,\Gamma \rangle = 0$, 
    \item \textbf{Correlation:} $\langle\Gamma,(\OR_R\circ \THR^k_N) \rangle>1/3$,
    \item \textbf{Exponentially little mass on inputs of large Hamming weight:} $\sum_{x\notin (\pmone^{RN})^{\leq N}}|\Gamma(x)| \leq (2NR)^{-\tilde{\Omega}\left(R^{\frac{3}{4} - \frac{1}{4k}}\right)}$ for all $x\notin (\pmone^{RN})^{\leq N}$.
\end{itemize}

Next, Lemma~\ref{lem: RS} implies existence of a function $\nu$ that equals $\Gamma$ on $x\notin (\pmone^{RN})^{\leq N}$, has pure high degree $\tilde{\Omega}\left(R^{\frac{3}{4} - \frac{1}{4k}}\right)$, and $\|\nu\|_1 \leq 1/10$.
The function $\W : \pmone^{RN} \to \R$ defined by $\mathcal{W}(x) := \frac{\Gamma(x)-\nu(x)}{\norm{\Gamma-\nu}_1}$ then satisfies the conditions in Equations~\eqref{eqn: zero_mass},~\eqref{eqn: l1_norm},~\eqref{eqn: final correlation} and~\eqref{eqn: phd} (see Section~\ref{sec: finalproof} for proofs). Theorem~\ref{thm: mainsummary} then follows by Lemma~\ref{lem: dualitypromise} and Lemma~\ref{lem: adegamp}.

\medskip \noindent \textbf{Organization of the rest of this section and the proof of Theorem~\ref{thm: mainsummary}.} The rest of this section is devoted towards providing a sketch of how we construct such a dual witness $\Gamma$.
In the next subsection we first sketch an outline of the approximate degree lower bound in~\cite{BKT18}, and in the subsequent subsection we elaborate on where our approach differs from theirs. Section~\ref{sec: mainproof} presents auxiliary lemmas that will be used in the formal proof of Theorem~\ref{thm: mainsummary}, while Section~\ref{sec:endofmainproof} contains the proof itself.

\subsection{Prior Work}
At a high level, we follow the same outline as followed in~\cite{BKT18}, who exhibited a dual witness $\Lambda$ witnessing $\adeg(\DIST^k_{N,R}) = \tilde{\Omega}\left(R^{\frac{3}{4}-\frac{1}{2k}}\right)$ for the same ranges of $k, N, R$ that we consider. In this section we sketch their construction.  Their dual witness takes the form $\Lambda = \theta \star \phi \star \psi$, where $\theta, \phi, \psi$  each have $\ell_1$-norm 1 and additionally satisfy the properties below.
\begin{itemize}
    \item The function $\psi$ satisfies:
    \begin{itemize}
        \item The false positive error between $\THR^k_N$ and $\psi$ is $O(1/N)$.
        \item The false negative error between $\THR^k_N$ and $\psi$ is at most $\frac{1}{2} - \frac{2}{4^k}$. 
        \item The pure high degree of $\psi$ is $\tilde\Omega(\sqrt{R}N^{-1/(2k)})$.
        \item $\psi$ satisfies a ``weak decay condition'', viz.~$\sum_{|x| = t}|\psi(x)| \leq \sigma \exp(-\beta t)/t^2$ for some constant $\sigma$ (for general $k$, the value of $\sigma$ only depends on $k$), and $\beta = \tilde\Omega(\sqrt{R}N^{1/(2k)})$.
        \end{itemize}
    \item The function $\phi$ is defined on $4^k$ inputs, and is defined as in Claim~\ref{claim: halfhalf}.
    \item $\theta$ is constructed as in Claim~\ref{clm: canonical_dual_OR} with $n = R/4^k$.
\end{itemize}

The facts that $\|\Lambda\|_1 = 1$ and $\phd(\Lambda) = \tilde\Omega(R^{3/4}N^{-1/(2k)})$ follow immediately from the definitions of $\theta, \phi, \psi$, and the fact that dual block composition preserves $\ell_1$-norm and causes pure high degree to increase multiplicatively (Lemma~\ref{lem: dbproperties}).

Next they use the fact that dual block composition is associative (Equation~\eqref{eqn: dbcassociative}) to express $\Lambda$ as $(\theta \star \phi) \star \psi$ and conclude using Claim~\ref{clm: strongdualdecay} that $\Lambda$ places exponentially small (in $R^{\frac{3}{4} - \frac{1}{2k}}$) mass on inputs in $\pmone^{RN}$ of Hamming weight larger than $N$.

It remains to show the correlation bound, i.e.,~$\langle \Lambda, \OR_R \circ \THR^k_N \rangle > 1/3$.
For the correlation analysis it is convenient to view $\Lambda$ as $\theta \star (\phi \star \psi)$.  The following is the outline of their correlation analysis.
\begin{enumerate}
    \item By construction, $\delta^+_{\THR^k_N, \psi} = O(1/N)$ and $\delta^-_{\THR^k_N, \psi} \leq \frac{1}{2} - \frac{2}{4^k}$.
    \item By Lemma~\ref{lem: halfhalfamplification}, the false positive error between $\OR_{4^k} \circ \THR^k_N$ and $\phi \star \psi$ remains $O(1/N)$, whereas the the false negative error between $\OR_{4^k} \circ \THR^k_N$ and $\phi \star \psi$ becomes a small enough constant.
    \item 
    As mentioned in Section~\ref{lowerboundoverview},
    the very low $(-1)$-certificate complexity of $\OR_R$
    renders false-negative errors benign. Thus
    the false-negative and false-positive error rates achieved in the last bullet point are sufficient 
    to ensure $\langle \theta \star (\phi \star \psi), \OR_{R/4^k} \circ (\OR_{4^k} \circ \THR^k_N)\rangle \geq 1/3$ by showing $\langle \theta \star (\phi \star \psi), \OR_{R/4^k} \circ (\OR_{4^k} \circ \THR^k_N)\rangle \approx \langle \theta, \OR_{R/4^k} \rangle$.
\end{enumerate}

Roughly, where we improve over this prior work is in item 3 above. Whereas \cite{BKT18} needed a false-positive error rate for $\phi \star \psi$ of $O(1/N)$ to ensure
that their final dual witness $\Lambda$ is well-correlated with $\OR_R \circ \THR^k_N$, we modify the construction of $\Lambda$ so 
that a false-positive error rate of roughly $1/\sqrt{N}$
suffices to ensure good correlation of the final
dual witness with $\OR_R \circ \THR^k_N$. 

\subsection{Our Construction}

As in the previous section, our construction of $\Gamma$ is also based on three dual witnesses. The functions $\theta, \phi$ are exactly the same as in the previous section. Our $\psi$ is a fairly straightforward modification of the one described in the previous section, that has a larger pure high degree, at the cost of a worse false positive error. A little more formally, our functions $\theta, \phi, \psi$ have $\ell_1$-norm equal to 1, and additionally satisfy the following.

\begin{itemize}
    \item The function $\psi$ satisfies:

    \begin{itemize}
        \item The false positive error between $\THR^k_N$ and $\psi$ is $\tilde O(1/\sqrt{N})$.
        \item The false negative error between $\THR^k_N$ and $\psi$ is at most $\frac{1}{2} - \frac{2}{4^k}$.
        \item The pure high degree of $\psi$ is $\tilde\Omega(\sqrt{R}N^{-1/(4k)})$.
        \item $\psi$ satisfies a ``weak decay condition'', viz.~$\sum_{|x| = t}|\psi(x)| \leq \sigma \exp(-\beta t)/t^2$ for some constant $\sigma$ (for general $k$, the value of $\sigma$ only depends on $k$), and $\beta = \tilde\Omega(\sqrt{R}N^{1/(4k)})$.
        \end{itemize}
    \item The function $\phi$ is defined on $4^k$ inputs, and is defined as in Claim~\ref{claim: halfhalf}.
    \item $\theta$ is constructed as in Claim~\ref{clm: canonical_dual_OR} with $n = R/4^k$.
\end{itemize}

If we were to define $\Gamma = \theta \star \phi \star \psi$, all the analyses from the previous section would work, except for the correlation analysis, which fails. To fix this, our main technical contribution is to not use dual block composition, but rather a variant of it inspired by a result of Sherstov~\cite{SDPT}. Our function $\Gamma$ takes the form $\Gamma = \theta \bullet (\phi \star \psi)$, where $\bullet$ denotes our variant of dual block composition. In a little more detail, \begin{align*} &\Gamma(x_1, \dots, x_{R/4^k})   :=\\& \theta \bullet (\phi \star \psi)(x) = \frac{1}{p_{\eta}(1-2\epsilon^+,\dots, 1-2\epsilon^+)} \cdot (\theta \star(\phi\star\psi))(x_1, \dots, x_{R/4^k})\cdot p_{\eta}(\alpha(x_1),\dots,\alpha(x_{R/4^k})),\end{align*} for $$\epsilon^+ = \epsilon^+_{\phi\star\psi,\OR_{4^k}\circ\THR^k_N},$$ \[\epsilon^- = \epsilon^-_{\phi\star\psi,\OR_{4^k}\circ\THR^k_N},\] $\eta$ is a parameter 
that we set later, and $p_\eta$ and $\alpha$ 
are functions
whose definitions we elaborate on later in this section.

We first give a very high-level idea of how we prove the required properties of $\Gamma$, and then elaborate on the definitions of $\eta, p_\eta$ and $\alpha$.

\begin{itemize}
    \item \textbf{Normalization:} Following along similar lines as~\cite[Claim 6.2]{SDPT}, we prove that $\|\Gamma\|_1 = 1$ by modifying the proof that dual block composition preserves $\ell_1$-norm, crucially exploiting properties of $p_\eta$ and $\alpha$ (see Claim~\ref{clm: gamma_l1_norm}).
    \item \textbf{Pure high degree:} Using our definition of $p_\eta$, and $\alpha$, one can show (Claim~\ref{clm: dual_sherstov_tweak}) that the pure high degree of $\theta \bullet (\phi \star \psi)$ is at least $(\phd(\theta) - \eta)\phd(\phi \star \psi)$. The value of $\eta$ is chosen to be $\phd(\theta)/2$ so that this quantity is the same order of magnitude as $\phd(\theta)\phd(\phi \star \psi) = \phd(\theta)\phd(\psi)$, 
    which is $\tilde\Omega(R^{3/4}N^{-1/(4k)})$.
    \item \textbf{Exponentially little mass on inputs of large Hamming weight:} By a similar argument as sketched in the last section, it can be shown that the mass placed by $(\theta \star \phi) \star \psi$ on inputs of Hamming weight larger than $N$ is exponentially small in $\tilde\Omega(R^{\frac{3}{4} - \frac{1}{4k}})$. Since $\theta \bullet (\phi \star \psi) := \frac{1}{p_{\eta}(1-2\epsilon^+,\dots, 1-2\epsilon^+)} \cdot (\theta \star(\phi\star\psi))(x_1, \dots, x_{R/4^k})\cdot p_{\eta}(\alpha(x_1),\dots,\alpha(x_{R/4^k}))$, it suffices to show that the maximum absolute value of $\frac{p_{\eta}(\alpha(x_1),\dots,\alpha(x_{R/4^k}))}{p_{\eta}(1-2\epsilon^+,\dots, 1-2\epsilon^+)}$ is at most exponentially large in $R^{\frac{3}{4} - \frac{1}{4k}}$, which we do in Claim~\ref{clm: gamma_strong_decay}.
    \item \textbf{Correlation:}
    Conceptually, the function $p_\eta : \pmone^{R/4^{k}} \to \R$ can be viewed as one that ``corrects'' $\theta \star (\phi \star \psi)$: it ``counts'' the number of false positives fed to it by $\phi \star \psi$, and changes the output of $\theta \star (\phi \star \psi)$ to 0 on inputs where this number is any integer between 1 and $\eta$.  The function $\alpha : \pmone^N \to \R$ acts as the function that, in a sense, indicates whether or not $\phi \star \psi$ is making a \emph{false positive} error.  
    \begin{itemize}
        \item \textbf{Detecting errors:} The function $\alpha$ takes three possible output values: 
        it outputs $-1$ for $x\in E^{+}(\OR_{4^k}\circ\THR^k_N, \phi \star \psi)$ and outputs either $1$ or a value very close to $1$ for $x\notin E^{+}(\OR_{4^k}\circ\THR^k_N, \phi \star \psi)$.
        This definition of $\alpha$ is our biggest departure from Sherstov's construction in~\cite{SDPT}; Sherstov defined $\alpha$ to output $-1$ for \emph{both} false-positive and false-negative errors, whereas our $\alpha$ only outputs $-1$ for false-positive errors.
        \item \textbf{Zeroing out errors:} Define the function $p_\eta$ to be (the unique multilinear extension of) the function that outputs 0 if its input has Hamming weight between $1$ and $\eta$. 
        Recall that our construction considers the dual witness $$\frac{1}{p_{\eta}(1-2\epsilon^+,\dots, 1-2\epsilon^+)} \cdot (\theta \star(\phi\star\psi))(x_1, \dots, x_{R/4^k})\cdot p_{\eta}(\alpha(x_1),\dots,\alpha(x_{R/4^k})),$$
        and the purpose of multiplying
        $\theta \star (\phi \star \psi)$ by $p_{\eta}$
        is for $p_{\eta}$ to zero out most inputs
        in which one or more false-positive errors
        are being fed by $\phi \star \psi$
        into $\theta$ (see Equation \eqref{defn: dualblockcomposition}). 

        Unfortunately, $p_{\eta}$ is nonzero
        on inputs of Hamming weight more than $\eta$.
        Hence, in terms of the correlation analysis, 
        a key question that must be addressed
        is: what fraction of the
        $\ell_1$-mass of $\theta \star (\phi \star \psi)$
        is placed on inputs where more than
        $\eta$ copies of $\phi \star \psi$ make
        a false-positive error? We need
        this fraction to be very small, because multiplying
        by $p_{\eta}$
        fails to zero out such inputs. 
        
        Note that under the distribution defined by $|\phi \star \psi|$, the \emph{expected} number of false positive errors fed into $\theta$ is $(R/4^{k}) \cdot \epsilon^+$. Since we have set $\eta = O(\sqrt{R/(4 \cdot 4^k)})$, it suffices to have $\epsilon^+ \ll 1/(c\eta)$ for some large enough constant $c$ to conclude that with high probability (over the distribution $|\phi \star \psi|$), the number of false positive errors fed into $\theta$ is at most a small constant times $\eta$. It turns out that this value of $\epsilon^+$ is indeed attained by $\phi \star \psi$, since the false positive error between $\THR^k_N$ and $\psi$ was set to be $\tilde O(1/\sqrt{N}) = \tilde O(1/\sqrt{R})$ to begin with. Thus, with high probability, multiplying $\theta \star (\phi \star \psi)$ by $p_{\eta}$ successfully zeros out all but an exponentially
        small fraction of the errors made by $\theta \star (\phi \star \psi)$ that can be attributed to false-positive errors made by $\phi \star \psi$. This intuitive proof outline is formalized in Claim~\ref{clm: gamma_corr}, which in turn is a formalization of Equation~\eqref{approximateeq} that holds with the setting of parameters mentioned above.
    \end{itemize}
\end{itemize}

\section{Properties of Auxiliary Functions}\label{sec: mainproof}
\label{sec:startofmainproof}
Given any function $f : \pmone^m \to \pmone$ and $\psi : \pmone^m \to \R$, $\norm{\psi}_1=1$, let $\epsilon^+ = \epsilon^+_{f,\psi}$ and $\epsilon^- = \epsilon^-_{f,\psi}$ as defined in Equation~\eqref{eqn: epsilon}. Define the function $\alpha_{f, \psi} : \pmone^m \to \R$ as 
\begin{equation}\label{eqn: alpha}
\alpha_{f,\psi}(x) :=
\begin{cases*}
  1 =: a^{+} & if $\psi(x)f(x)>0, \psi(x)>0$ \\
  \frac{1-2\epsilon^+ - \epsilon^-}{1-\epsilon^-} =: a^{-}  & if $\psi(x)f(x)>0, \psi(x)<0$ \\
  -1    & if $\psi(x)f(x)<0, \psi(x)>0$\\
  1     & if $\psi(x)f(x)<0, \psi(x)<0$.
\end{cases*}
\end{equation}
For the remaining sections, for $z_i\in\pmone$, $a^{z_i}= a^+$ if $z_i=1$, and $a^{z_i}=a^-$ if $z_i = -1$.

\begin{claim}\label{clm: alphaexpectation}
For any integer $m > 0$, any functions $f:\pmone^m\to\pmone$ and $\psi:\pmone^m\to\R$ such that $\norm{\psi}_1=1$, let $\alpha = \alpha_{f,\psi} : \pmone^m \to \R$ be as defined in Equation~\eqref{eqn: alpha}. Then for any integer $n>0$, any $z$ in $\pmone^n$, and all $i\in[n]$, 
    \begin{equation}
        \E_{(x_1, \dots, x_n)\sim\mu_{z}}[\alpha(x_i)] = 1-2\epsilon^+_{f,\psi}.
    \end{equation}
\end{claim}
\begin{proof}
Let $\epsilon^+ = \epsilon^+_{f, \psi}$ and $\epsilon^- = \epsilon^-_{f, \psi}$.
\begin{align*}
    \E_{\mu_{z}}[\alpha(x_i)] &= \E_{\mu_{z_i}}[\alpha(x_i)] \tag*{by Equation~\eqref{eqn: mu_z_prod_dist}}\\
    & = \begin{cases}
    \epsilon^+ \cdot -1 + (1 - \epsilon^+) & \text{if~}z_i = 1\\
    \epsilon^- \cdot 1 + (1 - \epsilon^-)\frac{1 - 2\epsilon^+ - \epsilon^-}{1 - \epsilon^-} & \text{if~}z_i = -1
    \end{cases}\tag*{by Definition~\ref{def: mu_z} and Equation~\eqref{eqn: alpha}}\\
    &= 1-2\epsilon^+.
\end{align*}
\end{proof}

Consider any positive integer $m$, functions $f : \pmone^m \to \pmone$ and $\psi : \pmone^m \to \R$, and any integers $\eta < n$. By Claim~\ref{clm: alphaexpectation}, Equation~\eqref{eqn: mu_z_prod_dist} and the fact that $p_{\eta} : [-1, 1]^n \to \R$ as defined in Lemma~\ref{lem: sherstovpk} is multilinear, it holds for all $z \in \pmone^n$ that
\begin{equation}\label{eqn: alphaexpectationallindices}
    \E_{\mu_z}[p_{\eta}(\alpha(x_1),\dots,\alpha(x_n))] = p_{\eta}(1-2\epsilon^+_{f, \psi},\dots, 1-2\epsilon^+_{f, \psi}).
\end{equation}

Let 
\begin{equation}\label{eqn: b}
    b^{+} := 0, \quad b^{-}:=\frac{\epsilon^+_{f, \psi}}{1-\epsilon^-_{f, \psi}},
\end{equation}
and $a^{+}, a^-$ be as defined in Equation~\eqref{eqn: alpha}. For the remaining sections, for $z_i\in\pmone$, $b^{z_i}:= b^+$ if $z_i=1$ and $b^{z_i}:=b^-$ if $z_i = -1$. Then, by multilinearity of $p_{\eta}$ and Definition~\ref{def: prod_dist_Pi}, for any $i \in [n]$ and any $c_1, \dots, c_{i - 1}, c_{i + 1}, \dots, c_n \in [-1,1]$ we have
\begin{equation}\label{eqn: pkexpectationnew}
    \E_{w\sim\Pi(b^{z_i})}[p_{\eta}(c_1, \dots, c_{i-1}, w, c_{i+1}, \dots, c_n)] = p_{\eta}(c_1, \dots, c_{i-1}, a^{z_i}, c_{i+1}, \dots,c_n),
\end{equation}
since $1 - 2b^+ = 1 = a^+$ and $1 - 2b^- = \frac{1-\epsilon^-_{f, \psi} - 2\epsilon^+_{f, \psi}}{1-\epsilon^-_{f, \psi}} = a^-$.
We also obtain that
\begin{equation}\label{eqn: pkexpectation}
    \E_{(w_1, \dots, w_n)\sim\Pi(b^{z_1},\dots,b^{z_n})}[p_{\eta}(w_1, \dots, w_n)] = p_{\eta}(a^{z_1},\dots,a^{z_n}),
\end{equation}
by Lemma~\ref{lem: expectation_multilinear}. We now state the setting for our next few claims.

\textbf{Assumptions for Claim~\ref{clm: piece3(a)}, Claim~\ref{clm: piece3(b)}, Claim~\ref{clm: finalcorrelation}, Claim~\ref{clm: gamma_l1_norm}}: Let $m, n$ be any positive integers, $\eta < n$ be any even positive integer, and $f:\pmone^m\to\pmone$ be any function. 
 Let $\zeta:\pmone^n\to\R$ be such that $\langle\zeta, \OR_n \rangle>\delta$ and $\|\zeta\|_1 = 1$, and $\xi: \pmone^m\to\R$ be any function such that $\|\xi\|_1 = 1$ and $\phd(\xi) \geq 1$.
Let $p_{\eta}:\pmone^n\to\R$ be as defined in Lemma~\ref{lem: sherstovpk}, let $\alpha = \alpha_{f,\xi}:\pmone^m\to\R$ be as defined in Equation~\eqref{eqn: alpha}, and define the distribution $\mu_{\xi}$ over $\pmone^{nm}$ as in Equation~\eqref{eqn: mudef}. Let $\epsilon^+ = \epsilon^+_{f, \xi}$, $\epsilon^- = \epsilon^-_{f, \xi}$, $\epsilon = \epsilon^+ + \epsilon^-$, and $A=\binom{n}{\eta+1}\frac{(\epsilon^+)^{\eta+1}}{(1-\epsilon^+)^n}$.

\begin{claim}\label{clm: piece3(a)}

\begin{align}
&\zeta(1^n)\E_{x\sim\mu_{1^n}}[p_{\eta}(\alpha(x_1),\dots,\alpha(x_n))\OR(f(x_1),\dots,f(x_n))] \nonumber\\
\geq & p_{\eta}(1-2\epsilon^+,\dots, 1-2\epsilon^+)\left(\zeta(1^n)-|\zeta(1^n)|2A \right).
\end{align}
\end{claim}

\begin{claim}\label{clm: piece3(b)}
\begin{align}
    &\sum_{z\neq 1^n}\zeta(z)\E_{\mu_z}[p_{\eta}(\alpha(x_1),\dots,\alpha(x_n))\OR(f(x_1),\dots,f(x_n))] \nonumber\\
    \geq & p_\eta(1-2\epsilon^+,\dots, 1-2\epsilon^+)\left(\sum_{z\neq1^n}\zeta(z)\OR(z)-\left(2-2\left(\frac{1-\epsilon^+ - \epsilon^-}{1-\epsilon^+}\right)(1-A)\right)\sum_{z\neq1^n}|\zeta(z)| \right).
\end{align}
\end{claim}

\begin{claim}\label{clm: finalcorrelation}
If $A<1$, then,
\begin{align}
    \langle \OR\circ f, (\zeta\star\xi)(p_{\eta}\circ\alpha)\rangle 
    & \geq p_{\eta}(1-2\epsilon^+,\dots, 1-2\epsilon^+)\cdot \left(\delta-\left(2-2\left(\frac{1-\epsilon^+ - \epsilon^-}{1-\epsilon^+}\right)(1-A)\right)\right).
\end{align}
\end{claim}
We first prove Claim~\ref{clm: finalcorrelation} using Claim~\ref{clm: piece3(a)} and Claim~\ref{clm: piece3(b)}, and prove those claims later.
\begin{proof}
\begin{align*}
&\langle \OR\circ f, (\zeta\star\xi)(p_{\eta}\circ\alpha)\rangle
    = \sum_{x\in \pmone^{mn}} (\OR\circ f)(x)(\zeta\star\xi)(p_{\eta}\circ\alpha)(x)\\
    &= \sum_{x\in\pmone^{mn}}\OR(f(x_1),\dots,f(x_n))2^n\zeta\left(\sgn(\xi(x_1)),\dots,\sgn(\xi(x_n))\right)p_{\eta}(\alpha(x_1),\dots,\alpha(x_n))\prod_{i=1}^n|\xi(x_i)|\tag*{by Definition~\ref{defn: dualblockcomposition}}\\
    &=\sum_{z\in\pmone^n}\zeta(z)\left(\sum_{x : \sgn(\xi(x_i)) = z_i \forall i \in [n]}2^np_{\eta}(\alpha(x_1),\dots,\alpha(x_n))\OR(f(x_1),\dots,f(x_n))\prod_{i=1}^n|\xi(x_i)| \right)\\
    &= \sum_{z\in\pmone^n}\zeta(z)\E_{\mu_z}[p_{\eta}(\alpha(x_1),\dots,\alpha(x_n))\OR(f(x_1),\dots,f(x_n))]\tag*{by Definition~\ref{def: mu_z} and $\Pr_{x_i\sim\mu_\xi}[\sgn(x_i) = 1] = \Pr_{x_i\sim\mu_\xi}[\sgn(x_i) = -1] = 1/2$ since $\phd(\xi) \geq 1$} \\
    &\geq  p_{\eta}(1-2\epsilon^+,\dots, 1-2\epsilon^+)\left(\zeta(1^n)\OR(1^n) - 2|\zeta(1^n)|A + \sum_{z \neq 1^n}\zeta(z)\OR(z)  \right.\\
    & \left. -\left(2-2\left(\frac{1-\epsilon^+ - \epsilon^-}{1-\epsilon^+}\right)(1-A)\right)\sum_{z\neq 1^n}|\zeta(z)|\right) \tag*{by Claims~\ref{clm: piece3(a)}, \ref{clm: piece3(b)} and $\OR(1^n)=1$}\\
    &=  p_{\eta}(1-2\epsilon^+,\dots, 1-2\epsilon^+) \left(\sum_{z \in \pmone^n}\zeta(z)\OR(z) - 2|\zeta(1^n)|A - \left(2-2\left(\frac{1-\epsilon^+ - \epsilon^-}{1-\epsilon^+}\right)(1-A)\right)\sum_{z\neq 1^n}|\zeta(z)| \right)\\
    &\geq  p_{\eta}(1-2\epsilon^+,\dots, 1-2\epsilon^+) \left(\delta -\max\left\{2A,2-2\left(\frac{1-\epsilon^+ - \epsilon^-}{1-\epsilon^+}\right)(1-A)\right\} \right) \tag*{since $|\zeta(1^n)|+\sum_{z\neq 1^n}|\zeta(z)|=1$ and $\langle \zeta, \OR \rangle > \delta$}\\
    &\geq  p_{\eta}(1-2\epsilon^+,\dots, 1-2\epsilon^+)\left(\delta-\left(2-2\left(\frac{1-\epsilon^+ - \epsilon^-}{1-\epsilon^+}\right)(1-A)\right)\right),
\end{align*}
where the last inequality holds because $\left(2-2\left(\frac{1-\epsilon^+ - \epsilon^-}{1-\epsilon^+}\right)(1-A)\right)-2A = (1-A)\left(2 - 2\left(\frac{1 - \epsilon^+ - \epsilon^-}{1 - \epsilon^+}\right)\right) > 0$
since $\left(\frac{1-\epsilon^+ - \epsilon^-}{1-\epsilon^+}\right)<1$, and $A < 1$.
\end{proof}
Next we prove Claim~\ref{clm: piece3(a)}.
\begin{proof}[Proof of Claim~\ref{clm: piece3(a)}]
Recall that $\mu_{1^n}$ is the distribution $\mu_\xi$ conditioned on the event that $\sgn(\xi(x_i)) = 1$ for all $i \in [n]$. Note that for all $x_1, \dots, x_n$ in the support of $\mu_{1^n}$ such that $I[(f(x_1),\dots,f(x_n))=1^n]$ (which means $f(x_i) = 1$ for all $i \in [n]$), we have by the definition of $a^{z_i}$ in Equation~\eqref{eqn: alpha}, that $\alpha(x_i) = a^+$ for all $i \in [n]$. Hence,
\begin{align}
    &\E_{\mu_{1^n}}[p_\eta(\alpha(x_1),\dots,\alpha(x_n))I[(f(x_1),\dots,f(x_n))=1^n]]\nonumber\\
    =&  \Pr_{x \sim \mu_{1^n}}[(f(x_1), \dots,f(x_n)) = 1^n]p_{\eta}(a^{+},\dots,a^{+})\tag*{} \nonumber\\
    =&  \left(\prod_{i=1}^n\Pr_{\mu_{1}}[f(x_i)=1]\right)p_{\eta}(a^{+},\dots,a^{+}) \tag*{by Equation~\eqref{eqn: mu_z_prod_dist}} \nonumber\\
    =&  \left(\prod_{i=1}^n (1-\epsilon^{+})\right) \E_{w\sim\Pi(b^{+},\dots,b^{+})}[p_{\eta}(w)]\tag*{by Equation~\eqref{eqn: epsilon} and Equation~\eqref{eqn: pkexpectation}}\nonumber\\
    \geq&  (1 - \epsilon^+)^n \Pr_{w\sim\Pi(b^{+},\dots,b^{+})}[w = 1^n]p_\eta(1^n) \tag*{since $p_\eta(w) \geq 0$ for all $w \in \pmone^n$ by Lemma~\ref{lem: sherstovpk}}\\
    =&   (1-\epsilon^+)^np_{\eta}(1^n),  \label{eq piece1(a)}
\end{align}
where the last line follows by Definition~\ref{def: prod_dist_Pi} and Equation~\eqref{eqn: b}.
Next,
\begin{align}
    &|\E_{\mu_{1^n}}[p_{\eta}(\alpha(x_1),\dots,\alpha(x_n))\OR(f(x_1),\dots,f(x_n))]-p_{\eta}(1-2\epsilon^+,\dots, 1-2\epsilon^+)| \nonumber\\
    = & |\E_{\mu_{1^n}}[p_{\eta}(\alpha(x_1),\dots,\alpha(x_n))\left(\OR(f(x_1),\dots,f(x_n))-1 \right)]|\tag*{by Equation~\eqref{eqn: alphaexpectationallindices}}\nonumber\\
    = & 2\E_{\mu_{1^n}}[p_{\eta}(\alpha(x_1),\dots,\alpha(x_n))(1-I[(f(x_1),\dots,f(x_n))=1^n])] \nonumber\\
    \leq & 2\E_{\mu_{1^n}}[p_{\eta}(\alpha(x_1),\dots,\alpha(x_n))]-2(1-\epsilon^+)^np_{\eta}(1^n) \label{eqn: piece2(a)}
\end{align}
by Equation~\eqref{eq piece1(a)}. 
Hence, by Equation~\eqref{eqn: piece2(a)},
\begin{align}
     &\zeta(1^n)\E_{\mu_{1^n}}[p_{\eta}(\alpha(x_1),\dots,\alpha(x_n))\OR(f(x_1),\dots,f(x_n))] \nonumber\\
     \geq & \zeta(1^n)p_{\eta}(1-2\epsilon^+,\dots, 1-2\epsilon^+)-|\zeta(1^n)|\left(2\E_{\mu_{1^n}}[p_{\eta}(\alpha(x_1),\dots,\alpha(x_n))] - 2(1-\epsilon^+)^np_{\eta}(1^n) \right) \nonumber\\
     = & p_{\eta}(1-2\epsilon^+,\dots, 1-2\epsilon^+)\left( \zeta(1^n)-|\zeta(1^n)|\left(2-\frac{2(1-\epsilon^+)^np_{\eta}(1^n)}{p_{\eta}(1-2\epsilon^+,\dots, 1 - 2\epsilon^+)}\right)\right)\tag*{by Equation~\eqref{eqn: alphaexpectationallindices}}\nonumber \\
     \geq & p_{\eta}(1-2\epsilon^+,\dots, 1-2\epsilon^+) \left(\zeta(1^n)-|\zeta(1^n)|2A \right),
\end{align}
where the last inequality follows as we have by Equation~\eqref{eqn: p_eta_bound} and the fact that $p_\eta$ is non-negative on all Boolean inputs (Lemma~\ref{lem: sherstovpk}) that
\begin{align}
    \E_{\Pi(\epsilon^+,\dots,\epsilon^+)}[p_\eta(z)] &\leq p_\eta(1^n)(1 - \epsilon^+)^n(1+A),  
\end{align}
which by Lemma~\ref{lem: expectation_multilinear} implies that
\begin{equation}\label{eqn: 1-A_inequality}
    \frac{2p_\eta(1^n)(1-\epsilon^+)^n}{p_\eta(1-2\epsilon^+,\dots,1-2\epsilon^+)} \geq 2(1+A)^{-1} \geq 2(1 - A),
\end{equation}
for all $A\geq 0$.
\end{proof}
We now prove Claim~\ref{clm: piece3(b)}.
\begin{proof}[Proof of Claim~\ref{clm: piece3(b)}]

We first introduce some notation that we use in this proof.
For any $i \in [n]$ and $r \in [-1, 1]$, let $y^i(r)$ denote the $n$-bit string $(1 - 2\epsilon^+, \dots, 1 - 2\epsilon^+, r, 1 - 2\epsilon^+, \dots, 1 - 2\epsilon^+)$, where $r$ is in the $i$'th position. It is easy to verify from its definition that $p_\eta$ is symmetric on $\pmone^n$.
Hence, for any $i \in [n]$,
\begin{align}
    &p_{\eta}(y^i(1))(1 - \epsilon^+) = \E_{w\sim\Pi(\epsilon^+, \dots,\epsilon^+)}[p_{\eta}(w,1)](1-\epsilon^+) \tag*{by Lemma~\ref{lem: expectation_multilinear}} \\
    \geq&  \Pr_{\Pi(\epsilon^+,\dots, \epsilon^+)}\left[w=1^{n-1}\right]p_{\eta}(1^n)(1-\epsilon^+)\tag*{since $p_\eta(w, 1) \geq 0$ for all $w \in \pmone^{n-1}$ by Lemma~\ref{lem: sherstovpk}}\\
    =& p_{\eta}(1^n)(1-\epsilon^+)^n,  \label{eqn: piece0(b)}
\end{align}
where the last equality follows from Definition~\ref{def: prod_dist_Pi}.
For any $z\neq 1^n$, let $i \in [n]$ be an index such that $z_i=-1$.\footnote{The notation $i_{z}$ is more accurate, but we drop the dependence on $z$ to avoid clutter. The underlying $z$ will be clear from context.} Fix any $z \neq 1^n$.
Note that by Equation~\eqref{eqn: alpha} and Equation~\eqref{eqn: mu_z_prod_dist}, $\alpha(x_i) = a^-$ for all $x_i$ in the support of $\mu_{z_i}$ satisfying $f(x_i) = -1$. Hence,
\begin{align}
    & \E_{x\sim\mu_z}[p_{\eta}(\alpha(x_1), \dots,\alpha(x_n))I[f(x_i)=-1]] = \Pr_{\mu_z}[f(x_i)=-1]p_{\eta}(y^i(a^{-})) \tag*{by Claim~\ref{clm: alphaexpectation}}\nonumber\\
    =&  \Pr_{\mu_z}[f(x_i)=-1]\E_{w\sim\Pi(b^{-})}[p_{\eta}(y^i(w))] \tag*{by Equation~\eqref{eqn: pkexpectationnew}} \nonumber\\
    =&  (1 - \epsilon^-)\E_{w\sim\Pi(b^{-})}[p_{\eta}(y^i(w))] \tag*{by Definition~\ref{def: mu_z} the definition of $\epsilon^-$ in Equation~\eqref{eqn: epsilon}}\\
    \geq&  (1-\epsilon^-)\left(1-\frac{\epsilon^+}{1-\epsilon^-}\right)p_{\eta}(y^i(1)) \tag*{by Equation~\eqref{eqn: b} and Definition~\ref{def: prod_dist_Pi} and $p_\eta$ is non-negative on $\pmone^n$ (Lemma~\ref{lem: sherstovpk})}\\
    =&  (1-\epsilon^- - \epsilon^+)p_{\eta}(y^i(1)) \label{eqn: piece1(b)}.
\end{align}
Next,
\begin{align}
    & |\E_{\mu_z}\left[p_{\eta}(\alpha(x_1),\dots,\alpha(x_n))\left(\OR(f(x_1),\dots,f(x_n))-\OR(z) \right)\right]|\nonumber\\
   =&  |\E_{\mu_z}\left[p_{\eta}(\alpha(x_1),\dots,\alpha(x_n))\left( \OR(f(x_1),\dots,f(x_n))+1\right)\right]| \tag*{since $\OR(z)=-1, \forall z\neq 1^n$}\nonumber\\
   \leq&  2\E_{\mu_z}\left[p_{\eta}(\alpha(x_1),\dots,\alpha(x_n))(1-I[f(x_i)=-1])\right] \tag*{since $p_\eta$ is non-negative on $\pmone^n$ by Lemma~\ref{lem: sherstovpk}}\\
   \leq&  2\E_{\mu_z}\left[p_{\eta}(\alpha(x_1),\dots,\alpha(x_n))\right]-2(1-\epsilon^- - \epsilon^+)p_{\eta}(y^i(1)) \label{eqn: piece2(b)},
\end{align}
where the last inequality follows by next applying Equation~\eqref{eqn: piece1(b)}. Finally, \begin{align}
    & \sum_{z\neq 1^n}\E_{\mu_z}[p_{\eta}(\alpha(x_1), \dots,\alpha(x_n))\OR(f(x_1), \dots,f(x_n))]\zeta(z)\nonumber\\
    \geq&  \sum_{z\neq 1^n}\zeta(z)\OR(z)p_{\eta}(y^i(1 - 2\epsilon^+))\nonumber -\sum_{z\neq 1^n}|\zeta(z)|\left(2p_{\eta}(y^i(1-2\epsilon^+))-2(1-\epsilon^- - \epsilon^+)p_{\eta}(y^i(1)) \right)\tag*{by Equation~\eqref{eqn: piece2(b)} and Equation~\eqref{eqn: alphaexpectationallindices}} \nonumber\\
    =&  p_{\eta}(1-2\epsilon^+, \dots, 1-2\epsilon^+)\left(\sum_{z\neq 1^n}\zeta(z)\OR(z) -\left(2-2\left(\left(\frac{1-\epsilon^+ - \epsilon^-}{1-\epsilon^+}\right)\right)\frac{(1-\epsilon^+)p_{\eta}(y^i(1))}{p_{\eta}(y^i(1-2\epsilon^+))}\right)\sum_{z\neq 1^n}|\zeta(z)| \right)\\
    \geq&   p_{\eta}(1-2\epsilon^+,\dots, 1-2\epsilon^+)\left(\sum_{z\neq 1^n}\zeta(z)\OR(z)
    -\left(2-2\left(\left(\frac{1-\epsilon^+ - \epsilon^-}{1-\epsilon^+}\right)\right)\frac{(1-\epsilon^+)^np_{\eta}(1^n)}{p_{\eta}(y^i(1-2\epsilon^+))}\right)\sum_{z\neq 1^n}|\zeta(z)| \right) \tag*{by Equation~\eqref{eqn: piece0(b)}}\\
    \geq&  p_\eta(1-2\epsilon^+,\dots, 1-2\epsilon^+)\left(\sum_{z\neq1^n}\zeta(z)\OR(z)-\left(2-2\left(\left(\frac{1-\epsilon^+ - \epsilon^-}{1-\epsilon^+}\right)\right)(1-A)\right)\sum_{z\neq1^n}|\zeta(z)| \right). \tag*{ by Equation~\eqref{eqn: 1-A_inequality}}
\end{align}
\end{proof}

Finally, we require a closed form expression for $\|(\zeta\star\xi)(p_{\eta}\circ\alpha)\|_1$.
\begin{claim}\label{clm: gamma_l1_norm}
\begin{equation}
    \norm{(\zeta\star\xi)(p_{\eta}\circ\alpha)}_1 = p_{\eta}(1-2\epsilon^+,\dots, 1-2\epsilon^+).
\end{equation}
\end{claim}
The proof of the claim follows along the lines as that of~\cite[Claim 6.2]{SDPT}, but we provide the proof for completeness.
\begin{proof}
Consider the distribution $\mu$ on $\pmone^{mn}$ defined by $\mu(x_1, \dots, x_n) = \prod_{i = 1}^n \mu_\xi(x_i)$. Since $\phd(\xi) \geq 1$, we conclude that the string $\left(\sgn(\xi(x_1)), \dots, \sgn(\xi(x_n))\right)$ is uniformly distributed in $\pmone^n$ when $(x_1, \dots, x_n)$ is sampled from $\mu$. Hence, we have \begin{align}
    \norm{(\zeta\star\xi)(p_{\eta}\circ\alpha)}_1 &= \sum_{x\in\pmone^{mn}}2^n\zeta(\sgn(\xi(x_1)),\dots,\sgn(\xi(x_n)))p_{\eta}(\alpha(x_1),\dots,\alpha(x_n))\prod_{i=1}^n|\xi(x_i)| \nonumber\\
    &=\sum_{z\in\pmone^n}|\zeta(z)\E_{\mu_z}[p_{\eta}(\alpha(x_1),\dots,\alpha(x_n))]| \nonumber\\
    & = \sum_{z\in\pmone^n}|\zeta(z)|\E_{\mu_z}[p_{\eta}(\alpha(x_1),\dots,\alpha(x_n))] \tag*{since $p_\eta$ is non-negative on $\pmone^n$ by Lemma~\ref{lem: sherstovpk}}\\
    &= p_{\eta}(1-2\epsilon^+,\dots, 1-2\epsilon^+)\sum_{z\in\pmone^n}|\zeta(z)| \tag*{by Equation~\eqref{eqn: alphaexpectationallindices}}\\
    &= p_{\eta}(1-2\epsilon^+,\dots, 1-2\epsilon^+).\tag*{since $\|\zeta\|_1 = 1$}
\end{align}
\end{proof}

\begin{claim}\label{clm: dual_sherstov_tweak}
Let $\Psi:\pmone^n\to\R$, $\Lambda:\pmone^m\to\R$, and $f : \pmone^m \to \R$ be any functions. For any positive integer $\eta$, let $\alpha = \alpha_{f,\Lambda}:\pmone^m\to\R$ be as defined in Equation~\eqref{eqn: alpha}, and $p_{\eta}:\pmone^n\to\R$ defined in Lemma~\ref{lem: sherstovpk}. Then
\begin{equation}
    \phd(\left(\Psi\star\Lambda\right)\cdot \left(p_{\eta}\circ\alpha)\right) > (\phd(\Psi)-\eta)\cdot\phd(\Lambda).
\end{equation}
\end{claim}
The proof follows along the same lines as that of~\cite[Equation (6.7)]{SDPT} and we omit it.

\section{Proof of Theorem~\ref{thm: mainsummary}}
\label{sec:endofmainproof}
Towards proving Theorem~\ref{thm: mainsummary}, it suffices to exhibit a dual polynomial (see Lemma~\ref{lem: dualitypromise}) that has $\ell_1$-norm 1, sufficiently large pure high degree, good correlation with $(\OR_R \circ \THR_{N}^k)^{\leq N}$, and places no mass outside $(\pmone^{RN})^{\leq N}$.
We first define a function $\Gamma$ (Definition~\ref{def: gamma}) that satisfies the first three properties above, and additionally satisfies a strong decay condition. In Section~\ref{sec: finalproof} we use $\Gamma$ to construct a dual polynomial $\W$, via Lemma~\ref{lem: RS}, satisfying all the requisite properties.
We now set several key variables.
\begin{itemize}
\item Let $R$ be sufficiently large and fix $k \leq (\log R)/4$. Set $T=\sqrt{R}$, $\eta=\left(\frac{c}{2}\sqrt{\frac{R}{4^k}}\right)-1$ where $c \in (0, 1]$ is the constant from Claim~\ref{clm: canonical_dual_OR} (assume without loss of generality that $\eta$ is even), $\sigma=(2k)^k$, $c_1, c_2 \in (0, 1]$ are constants fixed in the next bullet point, $\beta=\frac{c_2}{\sqrt{4^kkTN^{1/(2k)}\log N}}, \Delta=\frac{\beta\sqrt{\sigma}R}{4\ln^2R}=\frac{c_2 R}{4\ln^2 R}\sqrt{\frac{(2k)^k}{4^kkTN^{1/(2k)} \log N}}, N = \ceil{20\sqrt{\sigma}R}$. 
\item Let $\omega_T: [T]\cup\{ 0\}\to\R$ be a function that satisfies the conditions in Claim~\ref{clm: psiconstruction} and let $c_1, c_2$ be the constants for which the claim holds. Let $\psi:\pmone^N\to\R$ be defined by $\psi(x)=\omega_T(|x|)/\binom{N}{|x|}$ if $|x| \leq T$, and $0$ otherwise.
\item Let $\theta:\pmone^{R/4^k}\to\R$ be any function satisfying the conditions in Claim~\ref{clm: canonical_dual_OR} for $n = R/4^k$ (note that $R/4^k > 0$ since $k < (\log R)/2$). 
\item Let $\phi:\pmone^{4^k}\to\R$ be the function defined in Claim~\ref{claim: halfhalf} with $n = 4^k$.
\item Let $p_{\eta}:[-1, 1]^{R/4^k}\to\R$ be as defined in Lemma~\ref{lem: sherstovpk}.
\item Let $\alpha := \alpha_{\phi\star\psi,\OR_{4^k}\circ\THR^k_N}:\pmone^{4^kN}\to\R$ be as defined in Equation~\eqref{eqn: alpha}.
\item Let $\epsilon^+ := \epsilon^+_{\phi\star\psi,\OR_{4^k}\circ\THR^k_N}, \epsilon^- := \epsilon^-_{\phi\star\psi,\OR_{4^k}\circ\THR^k_N}$, and $\epsilon := \epsilon^+ + \epsilon^-$.
\end{itemize}

We next define the function $\Gamma$. 

\begin{definition}\label{def: gamma}
Let $\Gamma: \pmone^{NR}\to\R$ be defined by
\begin{align}
\Gamma(x_1, \dots, x_{R/4^k}) & := \frac{(\theta \star(\phi\star\psi))(x_1, \dots, x_{R/4^k})\cdot p_{\eta}(\alpha(x_1),\dots,\alpha(x_{R/4^k}))}{p_{\eta}(1-2\epsilon^+,\dots, 1-2\epsilon^+)},
\end{align}
where each $x_i\in\pmone^{4^kN}$.
\end{definition}

\subsection{Properties of $\Gamma$}

We now show in Section~\ref{sec: phd}, Section~\ref{sec: corr} and Section~\ref{sec: decay} that $\Gamma$ satisfies the following four properties.
\begin{itemize}
\item $\langle \Gamma, \OR_R \circ \THR^k_N \rangle > 1/3$.
\item $\|\Gamma\|_1 = 1$.
\item $\phd(\Gamma) = \Omega\left(\frac{1}{4^k k^2}\cdot\frac{1}{\sqrt{\log R}}\cdot R^{\frac{3}{4}-\frac{1}{4k}}\right)$.
\item $\sum_{x\notin (\pmone^{RN})^{\leq N}} |\Gamma(x)|\leq (2NR)^{-2(\Delta-\sqrt{R})}.$
\end{itemize}

\subsubsection{Pure High Degree}\label{sec: phd}
In this section we show the required lower bound on $\phd(\Gamma)$.
\begin{claim}\label{clm: finalphd}
\begin{equation}
    \phd(\Gamma) = \Omega\left(\frac{1}{4^k k^2}\cdot\frac{1}{\sqrt{\log R}}\cdot  R^{3/4-1/(4k)}\right).
\end{equation}
\end{claim}
\begin{proof}
\begin{align*}
\phd(\Gamma) & = \phd((\theta \star (\phi \star \psi))(p_\eta \circ \alpha))\tag*{by Definition~\ref{def: gamma}}\\
& \geq (\phd(\theta) - \eta) \cdot \phd(\phi \star \psi)\tag*{by Claim~\ref{clm: dual_sherstov_tweak}, using $\Psi = \theta, \Lambda = \phi \star \psi$, and $f = \OR_{4^k} \circ \THR_N^k$}\\
& \geq \frac{c}{2}\sqrt{\frac{R}{4^{k}}} \cdot \phd(\phi) \cdot \phd(\psi)\tag*{by Claim~\ref{lem: dbproperties}, Claim~\ref{clm: canonical_dual_OR}, and since $\eta = \left(\frac{c}{2}\sqrt{\frac{R}{4^{k}}}\right) - 1$}\\
& \geq \frac{c}{2}\sqrt{\frac{R}{4^{k}}} \cdot c_1\sqrt{4^{-k}k^{-1}TN^{-1/(2k)} \log^{-1} N} \tag*{by Claim~\ref{claim: halfhalf} and Equation~\eqref{eqn:thr-phd}}\\
& = \frac{cc_1}{2}\sqrt{\frac{RT}{4^{2k}kN^{1/(2k)} \log N}}\\
& = \frac{cc_1}{2} \sqrt{\frac{1}{\log N}} \sqrt{\frac{1}{4^{2k}k}} \sqrt{\frac{RT}{N^{1/2k}}}\\
& = \frac{cc_1}{2} \sqrt{\frac{1}{\log{20} + (1/2)\cdot k\log(2k)+\log R}} \sqrt{\frac{1}{4^{2k}k}}  \sqrt{\frac{R\sqrt{R}}{(20\sqrt{(2k)^k}R)^{1/2k}}} \tag*{using $T=\sqrt{R}$ and $N=20\sqrt{(2k)^k}R$}\\
& \geq  \frac{cc_1}{2^{9/8}}\cdot\frac{1}{\sqrt{k\log R}}\cdot \frac{1}{4^k\cdot 20^{1/(4k)}\cdot k^{5/8}}\cdot R^{3/4-1/(4k)} \tag*{since $k\log R > \log{20} + 1/2\cdot k\log(2k)+\log R$ for sufficiently large $R$} \\
& =  \frac{cc_1}{2^{9/8}} \cdot \frac{1}{4^k\cdot 20^{1/(4k)}\cdot k^{9/8}} \cdot \frac{1}{\sqrt{\log R}} \cdot R^{3/4-1/(4k)} \\
& \geq \frac{cc_1}{60}\cdot \frac{1}{\sqrt{\log R}}\cdot \frac{1}{4^k k^2}\cdot R^{3/4-1/(4k)}\tag*{since $1/20^{1/(4k)}>1/20$, for all $k\geq 2$}.
\end{align*}

\end{proof}

\subsubsection{Correlation}\label{sec: corr}

We first show that the function $\phi \star \psi$ has large correlation with $\OR_{4^k} \circ \THR^k_N$, the following analysis is essentially the same as in~\cite[Proposition 55]{BKT18Arxiv}.
\begin{claim}\label{clm: epsphipsi}
\begin{align}
\epsilon^+_{\OR_{4^k}\circ\THR^k_N, \phi\star\psi}  & \leq \frac{1}{24\sqrt{R} \log R},\label{eqn: epsilon_upper_bound}\\
\epsilon^-_{\OR_{4^k}\circ\THR^k_N, \phi\star\psi} & \leq e^{-4}.
\end{align}
\end{claim}

\begin{proof}
\begin{align*}
\epsilon^+_{\OR_{4^k}\circ\THR^k_N, \phi\star\psi} & = 2\delta^+_{\OR_{4^k}\circ\THR^k_N, \phi\star\psi} \tag*{by Claim~\ref{clm: epsilondeltarelation}}\\
& \leq 2 \cdot 4^k \cdot \delta^+_{\THR^k_N, \psi} \tag*{by Equation~\eqref{eqn: delta_plus_amplification}, using $M=4^k$}\\
& \leq \frac{1}{24\sqrt{N} \log N} \tag*{by Equation~\eqref{eqn:thr-pos-corr}}\\
& \leq \frac{1}{24\sqrt{R} \log R}. \tag*{since $N = \lceil 20 \sqrt{\sigma}R\rceil > R$} 
\end{align*}
Next,
\begin{align*}
\epsilon^-_{\OR_{4^k}\circ\THR^k_N, \phi\star\psi} & = 2\delta^-_{\OR_{4^k}\circ\THR^k_N, \phi\star\psi} \tag*{by Claim~\ref{clm: epsilondeltarelation}}\\
& \leq (2\delta^-_{\THR^k_N, \psi})^{4^k} \tag*{by Equation~\eqref{eqn: delta_minus_amplification} using $M = 4^k$}\\
& \leq \left(1 - \frac{4}{4^k}\right)^{4^k} \tag*{by Equation~\eqref{eqn:thr-neg-corr}}\\
& \leq e^{-4}. \tag*{since $(1 - 1/n)^n \leq 1/e$ for all $n \geq 1$}
\end{align*}
\end{proof}

\begin{claim}\label{clm: gamma_corr}
The function $\Gamma$ satisfies
\begin{align*}
    \|\Gamma\|_1 & = 1,\\
    \langle \Gamma, (\OR_R \circ \THR^k_{N}) \rangle & > 1/3.
\end{align*}
\end{claim}
\begin{proof}
The conditions of Claim~\ref{clm: gamma_l1_norm} are satisfied with $n = R/4^k, m = 4^kN, f = \OR_{4^k} \circ \THR^k_N, \zeta = \theta, \xi = \phi \star \psi, \eta = \left(\frac{c}{2}\sqrt{\frac{R}{4^{k}}}\right) - 1$.
Hence by Claim~\ref{clm: gamma_l1_norm},
\[
\|\Gamma\|_1 = \frac{\|(\theta \star (\phi \star \psi))(p_\eta \circ \alpha)\|_1}{p_\eta(1 - 2\epsilon^+, \dots, 1 - 2\epsilon^+)} = 1.
\]

Define $A=\binom{R/4^k}{\eta+1}\frac{(\epsilon^+)^{\eta+1}}{(1-\epsilon^+)^{R/4^k}}$. If $A < 1$, then the conditions of Claim~\ref{clm: finalcorrelation} are satisfied with the same parameters mentioned in the beginning on this proof.

We first show that $A < 1$, and then invoke Claim~\ref{clm: finalcorrelation}. To avoid clutter, define $\gamma = \eta + 1 = \frac{c}{2}\sqrt{\frac{R}{4^{k}}}$. 

\begin{align}
    A & = \binom{R/4^k}{\gamma} \frac{(\epsilon^+)^\gamma}{(1-\epsilon^+)^{R/4^k}} \nonumber\\
    & \leq  \left(\frac{Re}{4^k\gamma}\right)^\gamma \left(\frac{1}{24\sqrt{R}\log R}\right)^\gamma \left(1-\frac{1}{24\sqrt{R}\log R}\right)^{-R/4^k} \tag*{using Claim~\ref{clm: epsphipsi} and $\binom{m}{n}\leq \left(\frac{me}{n}\right)^n$}\nonumber\\
    & \leq \left(\frac{e}{24} \right)^\gamma \left(\frac{\sqrt{R}}{4^k \gamma \log R} \right)^\gamma \cdot3^{\sqrt{R}/(4^k24 \log R)} \tag*{rearranging terms and using $(1-1/n)^n\leq 1/e$ for $n>0$} \nonumber\\
    & = \left(\frac{e}{12} \right)^\gamma \left(\frac{1}{c\sqrt{4^k}\log R} \right)^\gamma \cdot3^{\sqrt{R}/(4^k24\log R)} \tag*{since $\gamma=\frac{c}{2}\sqrt{R/4^k}$} \nonumber\\
    & = \left(\frac{e \cdot 3^{1/(12c\sqrt{4^k} \log R)}}{12c\sqrt{4^k}\log R} \right)^\gamma \tag*{again using $\gamma=\frac{c}{2}\sqrt{R/4^k}$}\nonumber\\
    & \leq (e/48)^\gamma \tag*{since $3^{\frac{1}{12c\sqrt{4^k}\log R}} \leq 2$ and $12c\sqrt{4^k}\log R \geq 8$ for sufficiently large $R$}\nonumber\\
    &\leq 1/16. \label{eqn: aupperbound}
\end{align}

Thus, the conditions in Claim~\ref{clm: finalcorrelation} are satisfied. By the definition of $\Gamma$, we have
\begin{align*}
\langle \Gamma, (\OR_R \circ \THR^k_{N}) \rangle & = \frac{\langle (\theta \star(\phi\star\psi))\cdot (p_{\eta}\circ\alpha), \OR_{R/4^k}\circ(\OR_{4^k}\circ\THR^k_{N}) \rangle}{p_\eta(1 - 2\epsilon^+, \dots, 1 - 2\epsilon^+)}\\
& \geq \delta-\left(2-2\left(\frac{1-\epsilon^+ - \epsilon^-}{1-\epsilon^+}\right)(1-A)\right) \tag*{by Claim~\ref{clm: finalcorrelation}}\\
& \geq 3/5 -\left(2\frac{\epsilon^-}{1 - \epsilon^+}+2A\left(\frac{1-\epsilon^+ - \epsilon^-}{1-\epsilon^+}\right)\right) \tag*{since $\delta \geq 3/5$ by Claim~\ref{clm: canonical_dual_OR}}\\
& \geq 3/5 - \left(2e^{-4} \frac{1}{1 - \frac{1}{24\sqrt{R} \log R}} + 2A\right) \tag*{by Claim~\ref{clm: epsphipsi} and $\frac{1-\epsilon^+ - \epsilon^-}{1 - \epsilon^+} < 1$}\\
& > 3/5 - 1/8 - 1/8 \tag*{by Equation~\eqref{eqn: aupperbound} and since $R$ is sufficiently large}\\
& > 1/3.
\end{align*}

\end{proof}

\subsubsection{Strong Decay}\label{sec: decay}

We first state and prove a property of $p_\eta$ that we require.

\begin{claim}\label{clm: p_eta_lowerbound}
\begin{equation}
    p_{\eta}(1-2\epsilon^+,\dots, 1-2\epsilon^+)> 1.
\end{equation}
\end{claim}
\begin{proof}
\begin{align}
    & p_{\eta}(1-2\epsilon^+,\dots, 1-2\epsilon^+) =\E_{w\sim\Pi(\epsilon^+,\dots, \epsilon^+)}[p_{\eta}(w)] \tag*{by Lemma~\ref{lem: expectation_multilinear}}\\
    &\geq \Pr_{\Pi(\epsilon^+,\dots, \epsilon^+)}\left[w=1^{R/4^k}\right]p_{\eta}(1^{R/4^k}) \tag*{since $p_\eta$ is non-negative on $\pmone^{R/4^k}$ by Lemma~\ref{lem: sherstovpk}}\\
    &= (1-\epsilon^+)^{R/4^k} \eta! \tag*{ by Equation~\eqref{eqn: p_eta_all_false}} \\
    &\geq \left(1-\frac{1}{24\sqrt{R}\log R}\right)^{R/4^k}\cdot2^{\frac{c}{2}\sqrt\frac{R}{4^k} - 1} \tag*{by Equation~\eqref{eqn: epsilon_upper_bound} and using $\eta = \frac{c}{2}\sqrt{\frac{R}{4^k}} - 1$} \\
    &> \left(\frac{1}{3}\right)^{\left(\sqrt{R}/(4^k24\log R)\right)} \cdot 2^{c\sqrt{R/(4^k \cdot 4)} - 1} \tag*{since $R$ is sufficiently large and $(1 - 1/n)^n < 1/e$ for $n>0$}\\
    &= \frac{2^{\sqrt{R}\left(\frac{c}{2^{k + 1}} - \frac{\log 3}{4^k 24\log R}\right)}}{2} > 2^{\sqrt{R}\left(\frac{c}{2^{k+2}}\right)-1} \geq 1,
\end{align}
since $R$ is sufficiently large and $k<(\log R)/4$.

\end{proof}

We next show that $\Gamma$ satisfies a particular decay property.
\begin{claim}\label{clm: gamma_strong_decay}
The function $\Gamma$ defined in Definition~\ref{def: gamma} satisfies
\begin{equation}\label{eqn: strong_decay}
    \sum_{x\notin (\pmone^{RN})^{\leq N}} |\Gamma(x)|\leq (2NR)^{-2(\Delta-\sqrt{R})}.
\end{equation}
\end{claim}
\begin{proof}
First note that by Definition~\ref{def: gamma} and Claim~\ref{clm: p_eta_lowerbound}, it suffices to show the same decay property for $(\theta \star(\phi\star\psi))\cdot (p_{\eta}\circ\alpha)(x)$, that is, $\sum_{x\notin (\pmone^{RN})^{\leq N}} |(\theta \star(\phi\star\psi))\cdot (p_{\eta}\circ\alpha)(x)|\leq (2NR)^{-2(\Delta-\sqrt{R})}$.

By associativity of dual block composition (Equation~\eqref{eqn: dbcassociative}), $\theta\star\phi\star\psi = (\theta\star\phi)\star\psi$. Recall that $\psi : \pmone^N \to \R$ is defined as $\psi(x)=\omega_T(|x|)/\binom{N}{|x|}$ if $|x| \leq T$, and $0$ otherwise, for $\omega_T$ satisfying the conditions in Claim~\ref{clm: psiconstruction}. Hence, $\psi$ satisfies the conditions of Claim~\ref{clm: final_psi_construction} and also those in Claim~\ref{clm: strongdualdecay}. Hence using Claim~\ref{clm: strongdualdecay} with $\Phi = \theta \star \phi$, we have
\begin{equation}\label{eqn: dbcdecay}
    \sum_{x\notin (\pmone^{RN})^{\leq N}}|((\theta\star\phi)\star\psi)(x)|\leq (2NR)^{-2\Delta}.
\end{equation}

For any $x\in\pmone^{RN}$, we write $x = (x_1,\dots,x_{R/4^k})$, where $x_i\in\pmone^{4^kN}$, for all $i$. 
\begin{align*}
    \sum_{x\notin (\pmone^{RN})^{\leq N}} |(\theta \star(\phi\star\psi))\cdot (p_{\eta}\circ\alpha)(x)| &=  \sum_{x\notin (\pmone^{RN})^{\leq N}}|((\theta\star\phi)\star\psi)(x)||p_{\eta}(\alpha(x_1),\dots,\alpha(x_{R/4^k}))| \\
    &\leq \max_{y\in[-1,1]^{R/4^k}}|p_{\eta}(y)|\sum_{x\notin (\pmone^{RN})^{\leq N}}|((\theta\star\phi)\star\psi)(x)| \tag*{since $\alpha(w) \in [-1, 1]$ for all $w \in \pmone^{4^kN}$ by Equation~\eqref{eqn: alpha}}\\
    & \leq (2NR)^{-2\Delta} \eta!\binom{R/4^k+\eta}{\eta} \tag*{by Claim~\ref{claim: pkupperbound} and Equation~\eqref{eqn: dbcdecay}}\\
    & \leq (2NR)^{-2\Delta} \left(c\sqrt{\frac{R}{4^{k+1}}}\right)!\left(\frac{2eR/4^k}{c\sqrt{R/4^{k+1}}}\right)^{c\sqrt{R/4^{k+1}}} \tag*{since $\eta = c\sqrt{R/4^{k+1}} - 1 < R/4^k$, and $\binom{a}{b} \leq (ae/b)^b$}\\
    & \leq (2NR)^{-2\Delta} \sqrt{R}^{\sqrt{R}}\left(\frac{8e}{c}\sqrt{\frac{R}{4^{k+1}}}\right)^{\sqrt{R/4^{k+1}}}\\
    & \leq(2NR)^{-2\Delta}(8eR/c)^{\sqrt{R}}\\
    & \leq (2NR)^{-2(\Delta-\sqrt{R})}. \tag*{since $R$ (and hence $N$) is sufficiently large}
\end{align*}
\end{proof}

\subsection{Final Dual Polynomial}\label{sec: finalproof}

We now prove Theorem~\ref{thm: mainsummary}.

\begin{proof}[Proof of Theorem~\ref{thm: mainsummary}]

We exhibit a function $\W : \pmone^{RN} \to \R$ satisfying

\begin{equation}\label{eqn: zero_mass}
\mathcal{W}(x)=0, \forall x\notin (\pmone^{RN})^{\leq N},
\end{equation}
\begin{equation}\label{eqn: l1_norm}
\norm{\mathcal{W}}_1=1
\end{equation}
\begin{equation}\label{eqn: final correlation}
\langle \mathcal{W}, (\OR_{R}\circ\THR_N^k) \rangle > 7/33,
\end{equation}
\begin{equation}\label{eqn: phd}
\phd(\mathcal{W}) = \Omega\left(\frac{1}{4^k k^2}\cdot\frac{1}{\log^{5/2}R}\cdot R^{\frac{3}{4}-\frac{1}{4k}}\right).
\end{equation}

 The theorem then follows by Lemma~\ref{lem: dualitypromise} and Lemma~\ref{lem: adegamp}.
Towards the construction of such a $\W$, first note that by Claim~\ref{clm: gamma_strong_decay} and Lemma~\ref{lem: RS} there exists a function $\nu:\pmone^{RN}\to\R$ that satisfies the following properties.
\begin{equation}\label{eqn: nu'_cancellation}
    |x|>N\Rightarrow \nu(x)=\Gamma(x),
\end{equation}
\begin{equation}\label{eqn: nu'_phd}
    \phd(\nu)\geq2(\Delta-\sqrt{R})-1,
\end{equation}
\begin{equation}\label{eqn: nu'_l1_norm}
    \norm{\nu}_1 \leq \frac{1}{10}.
\end{equation}

Define $\mathcal{W}: \pmone^{RN}\to\R$ by
\begin{equation}\label{eqn: wdef}
    \mathcal{W}(x) := \frac{\Gamma(x)-\nu(x)}{\norm{\Gamma-\nu}_1}.
\end{equation}

For any $x \notin (\pmone^{RN})^{\leq N}$, we have $\W(x) = \frac{\Gamma(x) - \nu(x)}{\|\Gamma - \nu\|_1} = 0$ by Equation~\eqref{eqn: nu'_cancellation}. This justifies Equation~\eqref{eqn: zero_mass}.

Equation~\eqref{eqn: l1_norm} immediately follows from Equation~\eqref{eqn: wdef}.

To justify Equation~\eqref{eqn: final correlation}, we have
\begin{align*}
\langle \mathcal{W}, \OR_{R}\circ \THR^k_N \rangle & = \frac{1}{\norm{\Gamma-\nu}_1}\left(\langle \Gamma,\OR_{R}\circ \THR^k_N \rangle -\langle \nu,\OR_{R}\circ \THR^k_N \rangle \right) \tag*{by Equation~\eqref{eqn: wdef}}\\
&\geq \frac{1}{\norm{\Gamma-\nu}_1} \left(1/3- \langle \nu,\OR_{R}\circ \THR^k_N \rangle \right) \tag*{by Claim~\ref{clm: gamma_corr}}\\
& \geq \frac{1}{\norm{\Gamma-\nu}_1} \{1/3- \norm{\nu}_1 \} \geq \frac{1}{\norm{\Gamma-\nu}_1} \frac{7}{30} \tag*{by Equation~\eqref{eqn: nu'_l1_norm}}\\
& \geq \frac{7}{33}. \tag*{since $\norm{\Gamma-\nu}_1\leq \norm{\Gamma}_1 + \norm{\nu}_1 \leq \frac{11}{10}$ by the triangle inequality}
\end{align*}

We have from Equation~\eqref{eqn: wdef} that
\begin{align}
    \phd(\mathcal{W}) & = \phd\left(\frac{\Gamma(x)-\nu(x)}{\norm{\Gamma-\nu}_1}\right) \\
    & = \phd(\Gamma(x)-\nu(x)) \\
    & \geq \min \{\phd(\Gamma), \phd(\nu) \}.\label{eqn: min_nu_gamma}
\end{align}

From Equation~\eqref{eqn: nu'_phd} we have
\begin{align}
    \phd(\nu)&\geq 2(\Delta-\sqrt{R})-1 \\
    &= 2\left(\frac{c_2 R}{4\ln^2 R}\sqrt{\frac{(2k)^k}{4^k kTN^{1/(2k)}\log N}} - \sqrt{R}\right)-1 \tag*{substituting the value of $\Delta$}\\
    &\geq 2\left( \frac{c_2}{4} \cdot \frac{1}{\log^2 R \sqrt{\log N}} \cdot \left(\frac{k}{2}\right)^{k/2} \frac{1}{k^{1/2}} \cdot \frac{R^{3/4}}{N^{1/(4k)}}-\sqrt{R}\right) -1 \tag*{using $T=\sqrt{R}$ and $\ln R < \log R$}\\
    &\geq 2\left(\frac{c_2}{4} \cdot  \frac{1}{\log^2 R \sqrt{k\log R}} \cdot \left(\frac{k}{2}\right)^{k/2} \frac{1}{k^{1/2}} \cdot \frac{R^{3/4}}{20^{1/(4k)}2^{1/8}k^{1/8}R^{1/(4k)}}-\sqrt{R}\right)-1 \tag*{substituting the value of $N$ and using $k\log R > \log N$ for sufficiently large $R$}\\
    &= 2\left(\frac{c_2}{2^{17/8}} \cdot \frac{1}{\log^2 R \cdot \sqrt{\log R}} \cdot\left(\frac{k}{2}\right)^{k/2} \frac{1}{k^{9/8}\cdot 20^{1/(4k)}}\cdot R^{3/4-1/(4k)} - \sqrt{R}\right) -1 \\
    & \geq 2\left( \frac{c_2}{2^{26/8}}\cdot \frac{1}{\log^{5/2}R}\cdot \frac{1}{20^{1/(4k)}} \cdot R^{3/4-1/(4k)} - \sqrt{R}\right)-1 \tag*{since $\left(\frac{k}{2}\right)^{k/2}\frac{1}{k^{9/8}}\geq \frac{1}{2^{9/8}}$ for all $k\geq 2$}\\
    & \geq 2\left(\frac{c_2}{320}\cdot \frac{1}{\log^{5/2}R}\cdot R^{3/4-1/(4k)} - \sqrt{R}\right) - 1\\
    & \geq \frac{c_2}{320}\cdot \frac{1}{\log^{5/2}R}\cdot R^{3/4-1/(4k)}-1 \tag*{since $\sqrt{R} \leq \frac{c_2}{640}\cdot \frac{1}{\log^{5/2}R}\cdot R^{3/4-1/(4k)}$ for $k \geq 2$, for sufficiently large $R$}\\
    & = \Omega\left(\frac{1}{\log^{5/2}R}\cdot R^{3/4-1/(4k)}\right).
\end{align}

Therefore by Claim~\ref{clm: finalphd} and Equation~\eqref{eqn: min_nu_gamma}, we have $\phd(\mathcal{W}) = \Omega\left(\frac{1}{4^k k^2}\cdot\frac{1}{\log^{5/2}R}\cdot R^{\frac{3}{4}-\frac{1}{4k}}\right)$, justifying Equation~\eqref{eqn: phd} and finishing the proof. 

\end{proof}

\section{An Upper Bound}
\label{sec:upperbound}
We extend ideas from Sherstov's upper bound on the approximate degree of surjectivity~\cite{algopoly} to prove
an approximate degree upper bound for $k$-distinctness, where $k$ is not necessarily a constant. 
We first note that it suffices to show an approximate degree upper bound on $\left(\OR_{R} \circ \THR^k_N\right)^{\leq N}$.

\begin{claim}\label{clm: upper_bound_connection}
For any positive integers $k, R, N$,
\begin{equation}
    \adeg\left(\DIST^k_{N,R}  \right) \leq \adeg\left(\left(\OR_{R} \circ \THR^k_N\right)^{\leq N}\right)\cdot O(\log R).
\end{equation}
\end{claim}

Claim \ref{clm: upper_bound_connection} 
has essentially appeared in multiple prior works, e.g., \cite[Equation 4]{BT17}, \cite[Section 3.4.1]{BKT18Arxiv}, \cite[Section 6]{algopoly}.
Claim \ref{clm: upper_bound_connection} is a converse
to Claim \ref{clm: dist_connection},
but is far more straightforward to prove
than Claim \ref{clm: dist_connection}. 
Claim \ref{clm: upper_bound_connection} follows from the fact that
$\DIST^k_{N,R}$ can be written as
an $\OR$ over all $R$ range items $j$
of the function that tests
whether $k$ or more copies of $i$
appear in the input list.
In more detail, for $i \in [N]$
and $j \in [R]$, let $y_{i,j}(x)=-1$ if the $i$th  item of the input list equals range item $j$. Note that $y_{ij}(x)$
is a function of degree at most $\lceil \log_2 R\rceil$ in $x$. Moreover, $$\DIST^k_{N,R} = (\OR_{R} \circ \THR^k_N)(y_{1, 1}(x), y_{2, 1}(x), \dots, y_{R, N}(x)).$$
Claim \ref{clm: upper_bound_connection} follows.

The following is our main theorem in this section.

\begin{theorem}\label{thm: upper}
For any positive integers $k, R, N$, with $k\leq N/2$,
\[
\adeg\left(\left(\OR_{R} \circ \THR^k_N\right)^{\leq N}\right) = O(N^{1/2}R^{1/4} \sqrt{k \log N}).
\]
\end{theorem}

For any integers $N \geq i \geq 0$, define the function $\exact^i_N : \pmone^N \to \zone$ by 
\[
\exact^i_N(x) = \begin{cases}
1 & |x| = i\\
0 & \text{otherwise}.
\end{cases}
\]

Note that 
\begin{equation}\label{eqn: exact}
\exact^i_N(x) = \sum_{S \subseteq [N] : |S| = i} \prod_{u \in S}\left(\frac{1 - x_u}{2}\right) \prod_{v \notin S}\left(\frac{1 + x_v}{2}\right).
\end{equation}

Recall that for integers $N \geq k \geq 0$, the function $\THR^k_N : \pmone^N \to \pmone$ is defined by
\[
\THR^k_N(x) = \begin{cases}
-1 & |x| \geq k\\
1 & \text{otherwise}.
\end{cases}
\]

We have
\begin{equation}\label{eqn: thr}
\THR^k_N(x) = 2\left(\sum_{i = 0}^{k-1} \exact^i_N(x)\right) - 1
\end{equation}
since exactly one summand outputs 1 if the Hamming weight of $x$ is less than $k$, and all summands output 0 otherwise.

For integers $m, R \geq 0$, define a degree-$m$ polynomial $p : \pmone^R \to \R$ by
\begin{equation}\label{eqn: chebOR}
p(x) = \frac{2}{T_m\left(1 + \frac{1}{R}\right)} \cdot T_m\left(\frac{\sum_{i = 1}^R x_i}{R} + \frac{1}{R}\right) - 1.
\end{equation}
Note that when $|x| = 0$, we have $\sum_{i = 1}^R x_i = R$, and hence $p(x) = 1$. When $|x| > 0$, we have $\frac{\sum_{i = 1}^R x_i}{R} + \frac{1}{R} \in [-1, 1]$, and by Equation~\eqref{eqn: abs1} this implies $T_m\left(\frac{\sum_{i = 1}^R x_i}{R} + \frac{1}{R}\right)  \in [-1, 1]$, and thus $p(x) \in \left[-1 - \frac{2}{1 + (m^2/R)}, -1 + \frac{2}{1 + (m^2/R)}\right]$ by Equation~\eqref{eqn: largeoutside1}. The next claim immediately follows.

\begin{claim}\label{claim: chebOR}
The degree-$m$ polynomial $p$ defined in Equation~\eqref{eqn: chebOR} uniformly approximates $\OR_R$ to error $\frac{2}{1 + (m^2/R)}$.
\end{claim}

We are now ready to prove our final upper bound.

\begin{proof}[Proof of Theorem~\ref{thm: upper}]
Let $m \geq 1$ be an integer parameter to be fixed later and let $T_m$ be the degree-$m$ Chebyshev polynomial.  Thus by Claim~\ref{claim: chebOR}, the function $\OR_R \circ \THR^k_N$ is approximated pointwise to error $\frac{2}{1 + (m^2/R)}$ by the degree-$m$ polynomial $p : \pmone^{RN} \to \R$ defined by

\begin{align*}
p(x) & = \frac{2}{T_m\left(1 + \frac{1}{R}\right)} \cdot T_m\left(\frac{1}{R} + \frac{1}{R}\sum_{j = 1}^R \THR^k_N (x_{j, 1}, \dots, x_{j, N})\right) - 1\\
& = \frac{2}{T_m\left(1 + \frac{1}{R}\right)} \cdot T_m\left(\frac{1}{R} -1 + \frac{2}{R} \sum_{j = 1}^R \sum_{i = 0}^{k - 1}\exact^i_N(x_{j, 1}, \dots, x_{j, N})\right) - 1 \tag*{by Equation~\eqref{eqn: thr}}\\
& = \frac{2}{T_m\left(1 + \frac{1}{R}\right)} \cdot T_m\left(\frac{1}{R} -1 + \frac{2}{R} \sum_{j = 1}^R \sum_{i = 0}^{k - 1}\left(\sum_{S \subseteq [N] : |S| = i} \prod_{u \in S}\left(\frac{1 - x_{j, u}}{2}\right) \prod_{v \notin S}\left(\frac{1 + x_{j, v}}{2}\right)\right) \right). \tag*{by Equation~\eqref{eqn: exact}}
\end{align*}

For simplicity of notation, define 
\begin{equation}
    C_{j,S}:= \prod_{u \in S}\left(\frac{1 - x_{j, u}}{2}\right) \prod_{v \notin S}\left(\frac{1 + x_{j, v}}{2}\right).
\end{equation}

We next show an upper bound on $\rho(p)$ (recall
that $\rho(p)$ is the conjunction norm of $p$ defined in Definition \ref{def: norm_conjunction}).
\begin{align}
    \rho(p) & = \abs{\frac{2}{T_m(1+\frac{1}{R})}} \cdot \rho\left(T_m\left(\frac{1}{R} -1 + \frac{2}{R} \sum_{j = 1}^R \sum_{i = 0}^{k - 1}\left(\sum_{S \subseteq [N] : |S| = i} C_{j, S}\right)\right)\right) \tag*{by Equation~\eqref{eqn: rho_constant_mult}}\\
    &\leq 2\cdot 3^m\cdot \rho\left( \frac{1}{R} -1 + \frac{2}{R} \sum_{j = 1}^R \sum_{i = 0}^{k - 1}\left(\sum_{S \subseteq [N] : |S| = i} C_{j, S}\right)\right)^m \tag*{by Equation~\eqref{eqn: rho_circ}, Equation~\eqref{eqn: chebcoeffs}, and $T_m(1+\frac{1}{R})>1$} \\
    &\leq 2\cdot 3^m \cdot \left( \abs{\frac{1}{R}-1} + \rho\left(\frac{2}{R} \sum_{j = 1}^R \sum_{i = 0}^{k - 1}\left(\sum_{S \subseteq [N] : |S| = i} C_{j, S}\right) \right) \right)^m
    \tag*{by Equation~\eqref{eqn: rho_add} } \\
    &\leq 2\cdot 3^m \cdot \left( 1 +\frac{2}{R} \sum_{j = 1}^R\rho\left( \sum_{i = 0}^{k - 1}\left(\sum_{S \subseteq [N] : |S| = i} C_{j, S}\right) \right)  \right)^m \tag*{by Equation~\eqref{eqn: rho_constant_mult} and Equation~\eqref{eqn: rho_add}} \\
    &\leq 2\cdot 3^m\cdot\left(1+2\cdot k\binom{N}{k} \right)^m \tag*{by Equation~\eqref{eqn: rho_add} and $\rho(C_{j,S})$ is at most $1$}\\
    &\leq \left(c_1\cdot k \binom{N}{k} \right)^m, \label{eqn: rhop}
\end{align}

for some positive constant $c_1$. By Claim~\ref{claim: conjbddHW}, we have the following. For each conjunction $f$ there is a degree-$d$ polynomial $p_f$ such that $|p_f(x) - f(x)| \leq 2^{-c\cdot d^2/N}$ for all $x \in \left(\pmone^{RN}\right)^{\leq N}$ for some positive constant $c$. By construction, $\deg(p) = m$ and $\abs{p(x) - \left(\OR_{R} \circ \THR^k_N\right)(x)} \leq 2/\left(1 + \frac{m^2}{R}\right)$ for all $x \in \pmone^{RN}$. By the triangle inequality, we obtain that for any integers $m, d \geq 0$,
\begin{align}
    E\left(\left(\OR_{R} \circ \THR^k_N\right)^{\leq N}, d\right) &\leq \frac{2}{1 + \frac{m^2}{R}} + \rho(p)\cdot 2^{-c\cdot d^2/N} \nonumber\\
    &\leq \frac{2}{1 + \frac{m^2}{R}} + 2^{ m \log (c_1 k\binom{N}{k})} \cdot 2^{-c\cdot d^2/N} \tag*{by Equation~\eqref{eqn: rhop}}\\
    &\leq \frac{2}{7} + 2^{ \sqrt{6R}\log(c_1 k\binom{N}{k})-c\cdot d^2/N } \tag*{setting $m=\sqrt{6R}$} \\
    &\leq \frac{2}{7} + 2^{ \sqrt{6R}(\log(c_1)+\log( k)+\log(N^k))-c\cdot d^2/N } \tag*{since $\binom{N}{k}\leq N^k$}\\
    & \leq \frac{2}{7} + 2^{3\sqrt{6R}k\log(N)-cd^2/N} \tag*{for sufficiently large $N$}\\
    & \leq \frac{1}{3}.\tag*{for $d=\frac{4}{c}\cdot R^{1/4}\sqrt{N k \log N}$ }
\end{align}

Hence there is a polynomial of degree $\frac{4}{c}\cdot R^{1/4}\sqrt{N k \log N}$ that approximates $\left(\OR_{R} \circ \THR^k_N\right)^{\leq N}$ within error $1/3$, and the theorem follows.

\end{proof}

Combining Claim~\ref{clm: upper_bound_connection} and Theorem~\ref{thm: upper} immediately yields an upper bound on the approximate degree of $k$-distinctness.

\begin{corollary}\label{cor: upper}
For any positive integers $R, N$ and $k \leq N/2$,
\[
\adeg\left(\DIST^k_{N,R}  \right) = O(\sqrt{k} N^{1/2}R^{1/4} \log R\sqrt{\log N}).
\]
\end{corollary}

Recall (cf.~Corollary \ref{cor: ambrangereduction}) that
 Ambainis \cite{ambainissmallrange} 
 showed 
that, for all functions that are symmetric both with respect to range elements and with respect to domain elements,
the approximate degree is the same for all range sizes greater than or equal to $N$.  This implies that the upper bound in Corollary
\ref{cor: upper} can
be refined to 
\[
\adeg\left(\DIST^k_{N,R}  \right) = O(\sqrt{k} N^{1/2}\min\{N, R\}^{1/4} \log R\sqrt{\log N}).
\]

\section*{Acknowledgements}
JT and SZ are supported by the National Science Foundation CAREER award (grant CCF-1845125). JT is grateful to Robin Kothari for extremely useful suggestions and discussions surrounding Theorem \ref{thm:informalupper}, and to Mark Bun for essential discussions regarding Theorem~\ref{thm: mainsummary}. SZ would like to thank Yao Ji for several helpful conversations.

\bibliography{bibo}

\appendix

\section{A Dual Polynomial for Threshold Function}

In this section, we prove Claim~\ref{clm: psiconstruction}.
We require the following well-known combinatorial identity. For a proof, see, for example,~\cite{OS10}.
\begin{fact}\label{fact: combid}
Let $N \in \mathbb{N}$ and let $p : \R \to \R$ be any polynomial of degree less than $N$. Then,
\[
\sum_{i = 0}^N (-1)^i\binom{N}{i}p(i) = 0.
\]
\end{fact}

\begin{proof}[Proof of Claim~\ref{clm: psiconstruction}]

Let $E_+ := \{t \colon \omega(t) > 0, t \ge k\}$, and $E_{-} := \{t\colon \omega(t) < 0, t < k\}$. 
By normalizing, it suffices to construct a function $\omega : [T] \cup \bra{0} \to \R$ 
such that

\begin{align} & \sum_{t \in E_+} |\omega(t)| \le \frac{1}{48\cdot4^k \sqrt{N} \log{N}} \cdot \|\omega\|_1  \label{eqn:thr-sym-pos-corr-ref} \\
& \sum_{t \in E_-} |\omega(t)| \le \left(\frac{1}{2} - \frac{2}{4^k}\right) \cdot \|\omega\|_1  \label{eqn:thr-sym-neg-corr-ref} \\
& \text{For all univariate polynomials } q \colon \R \to \R\text{, } \nonumber \\
&\deg (q) < c_1\sqrt{4^{-k} k^{-1}  T  N^{-1/(2k)} \log^{-1}{N}} \implies \sum_{t=0}^T \omega(t)  q(t) = 0 \label{eqn:thr-sym-phd-ref} \\
&  |\omega(t)| \le \frac{(2k)^k \exp\left(-c_2 t / \sqrt{4^kk  T  N^{1/(2k)} \log{N}}\right)\|\omega\|_1}{t^2} \qquad \forall t = 1, 2, \dots, T.\label{eqn:thr-sym-decay-ref}
\end{align}

Let $\ell = 100k\lceil N^{1/(2k)}4^k \log{N}\rceil$, and let $m = \lfloor \sqrt{T/\ell} \rfloor$. 
Define the set
\[S = \{1, 2, \dots, k\} \cup \{\ell i^2 : 0 \le i \le m\}.\]
Note that $|S|=k+m+1\geq m = (1/10) \sqrt{4^{-k} k^{-1}  T  N^{-1/(2k)} \log^{-1}{N}}$.  Define the polynomial $\omega : [T] \cup \bra{0} \to \R$ by 
\[
\omega(t) =\frac{(-1)^{T-t-m+1}}{T!} \binom{T}{t} \prod_{r \in ([T]\cup\{0\}) \setminus S} (t - r).
\]

The signs are chosen so that $\omega(k) < 0$, because in the expression
\begin{align*}
    \omega(k) &= \frac{(-1)^{T-k-m+1}}{T!}\binom{T}{k} \prod_{r \in ([T]\cup\{0\}) \setminus S} (k - r),
\end{align*}
the number of terms in the product is $|([T]\cup\{0\}) \setminus S|=T-k-m$, and each term in the product is negative for $k = 0$.

Let $q$ be any univariate polynomial of degree less than $|S|-1$. Then,
\begin{align*}
    \sum_{t=1}^T \omega(t)q(t) &= \frac{(-1)^{T-m+1}}{T!}\sum_{t=1}^T (-1)^t\binom{T}{t}\prod_{r\in([T]\cup\{0\})\setminus S}(t-r)\cdot q(t) \tag*{since $(-1)^{-t}=(-1)^t$ for all integer $t$} \\
    &= \frac{(-1)^{T-m+1}}{T!}\sum_{t=1}^T (-1)^t\binom{T}{t} p(t) \tag*{where $p(t):=\prod_{r\in([T]\cup\{0\})\setminus S}(t-r)\cdot q(t)$}\\
    &= 0 \tag*{by Fact~\ref{fact: combid}}
\end{align*}
where we could use Fact~\ref{fact: combid} since $\deg(p)\leq\deg(q)+\deg(\prod_{r\in([T]\cup\{0\})\setminus S}(t-r))\leq |[T]\cup\{0\}|-|S|+\deg(q) < T+1-|S|+|S|-1=T$.
  
Since $|S|-1 = k + m > m = (1/10) \sqrt{4^{-k} k^{-1}  T  N^{-1/(2k)} \log^{-1}{N}}$, we conclude that $\omega$ satisfies Equation~\eqref{eqn:thr-sym-phd-ref} for $c_1 = 1/10$.
We now show that Equation~\eqref{eqn:thr-sym-decay-ref} holds. For $t = 1, \dots, k$, we have

\[
\frac{(2k)^k \exp\left(-c_2 t / \sqrt{4^k k  T N^{1/(2k)}\log{N}}\right)}{t^2} \ge \frac{(2k)^k \exp\left(-c_2 \sqrt{k}\right)}{k^2} \ge 1
\]
as long as $c_2 \le 1/2$ and $k \ge 2$. Since $|\omega(t)| \le \|\omega\|_1$, the bound holds for $t = 1, \dots, k$.

Next, note that $\omega(t) = 0$ for $t \notin S$. For $t\in S$, we have
\begin{align*}
    |\omega(t)| &= \frac{1}{T!}\cdot \frac{T!}{t!(T-t)!}\cdot\frac{ \prod_{r \in ([T]\cup\{0\}) \setminus S} |t - r|\cdot  \prod_{r \in S\setminus\{t\}} |t - r|}{\prod_{r \in S\setminus\{t\}} |t - r|}\\
    &= \frac{1}{t!(T-t)!}\cdot \frac{\prod_{r \in ([T]\cup\{0\}) \setminus \{t\}} |t - r|}{\prod_{r \in S\setminus\{t\}} |t - r|}\\
    &= \frac{1}{t!(T-t)!}\cdot \frac{t!(T-t)!}{\prod_{r \in S\setminus\{t\}} |t - r|} =\prod_{r \in S\setminus\{t\}} \frac{1}{ |t - r|}.
\end{align*}
Thus,
\[
|\omega(t)| = 
\begin{cases}
 \prod\limits_{r \in S\setminus \{t\}} \frac{1}{|t - r|} & \text{ for } t \in S, \\
0 & \text{ otherwise.}
\end{cases}
\]

For $t \in \{0, 1, \dots, k\}$, we observe that
\begin{equation} \label{eqn:head-bound}
\frac{|\omega(t)|}{|\omega(k)|} = \frac{k! \cdot \prod_{i = 1}^m (\ell i^2 - k)}{t! \cdot (k - t)! \cdot \prod_{i = 1}^m (\ell i^2 - t)} \le \binom{k}{t}.
\end{equation}
Meanwhile, for $t = \ell j^2$ with $j \ge 1$, we get
\begin{align*}
\frac{|\omega(t)|}{|\omega(k)|} &= \frac{k! \cdot \prod_{i = 1}^m (\ell i^2 - k)}{\prod_{i = 1}^k(\ell j^2 - i) \cdot \prod_{i \in \bra{[m]\cup\{0\}} \setminus \{j\}} |\ell i^2 - \ell j^2|} \\
&\le \frac{k! \cdot \prod_{i = 1}^m \ell i^2}{(\ell j^2 - k)^{k} \cdot \prod_{i \in \bra{[m]\cup\{0\}} \setminus \{j\}} \ell(i+j)|i-j|} \\
&= \frac{2\cdot k!}{(\ell j^2 - k)^{k}} \cdot \frac{(m!)^2}{(m+j)!(m-j)!}.
\end{align*}
The first factor is bounded above by
\[
\frac{2\cdot k!}{(\ell - k)^k j^{2k}}.
\]
Since $\ell \ge 2k$ by our choice of $\ell$, and $k \ge 2$, this expression is at most
\[\frac{k^k}{(\ell/2)^kj^4} = \frac{(2k)^k}{\ell^k \cdot j^4}.\]

We control the second factor by
\begin{align*}
\frac{(m!)^2}{(m+j)!(m-j)!} &= \frac{m}{m+j} \cdot \frac{m-1}{m+j-1} \cdot \ldots \cdot \frac{m-j+1}{m+1} \\
&\le \left(\frac{m}{m+j} \right)^j \\
&\le  \left(1 - \frac{j}{2m} \right)^j \\
&\le e^{-j^2/2m},
\end{align*}
where the last inequality uses the fact that $1 - x \le e^{-x}$ for all $x$. 
Hence,
\begin{equation} \label{eqn:tail-bound}
\frac{|\omega(\ell j^2)|}{|\omega(k)|} \le \frac{(2k)^k}{\ell^k \cdot j^4} \cdot e^{-j^2/2m}.
\end{equation}
This immediately yields
\[\frac{|\omega(\ell j^2)|}{\|\omega\|_1} \le \frac{|\omega(\ell j^2)|}{|\omega(k)|} \le \frac{(2k)^k}{(\ell j^2)^2} \cdot e^{-\ell j^2/(2\ell m)},\]

If we choose $c_2 < 1/20$, we have $\frac{1}{2\ell m} \geq \frac{1}{2\sqrt{T\ell}}> \frac{c_2}{\sqrt{4^k kTN^{1/(2k)}\log{N}}}$ since $\ell = 100k\lceil N^{1/(2k)}4^k\log{N}\rceil$.
This establishes Equation~\eqref{eqn:thr-sym-decay-ref} for all $t = \ell j^2 > k$.  
Moreover, by Equation~\eqref{eqn:tail-bound},
\begin{align} 
\sum_{t > k} |\omega(t)| & \le |\omega(k)| \cdot \sum_{j = 1}^m \frac{(2k)^k}{\ell^k \cdot j^4} \cdot e^{-j^2/2m}\nonumber\\
&\le \left(\frac{2k}{\ell}\right)^k \cdot |\omega(k)| \cdot  \sum_{j=1}^m \frac{1}{j^4} \nonumber\\
& < \frac{(2k)^k}{(100k N^{1/(2k)}4^k \log N)^k} \cdot \frac{\pi^4}{50}\cdot |\omega(k)| \tag*{since $\sum_{j=1}^\infty 1/j^4 < \pi^4/50$} \\
&= \frac{1}{50^k \sqrt{N} (4^k)^k \log^k N}\cdot \frac{1}{50/\pi^4}\cdot |\omega(k)|\nonumber\\
&\leq  \frac{1}{50^2\sqrt{N} 4^k \log N} \cdot \frac{1}{50/\pi^4}\cdot |\omega(k)| \tag*{for all $k\geq 2$}\\
&\le \frac{|\omega(k)|}{48\cdot 4^k\cdot \sqrt{N} \log{N}}. \label{eqn:tail-sum}
\end{align}
Hence, since $\omega(k) < 0$, 
\[
\sum_{t \in E_+} |\omega(t)| \le \sum_{t > k} |\omega(t)| \le \frac{|\omega(k)|}{48\cdot 4^k\cdot \sqrt{N} \log{N}} \le \frac{\|\omega\|_1}{48\cdot 4^k\cdot \sqrt{N} \log{N}},
\]
which gives Equation~\eqref{eqn:thr-sym-pos-corr-ref}.

Finally, to establish Equation~\eqref{eqn:thr-sym-neg-corr-ref}, we combine Equation~\eqref{eqn:head-bound} and Equation~\eqref{eqn:tail-sum} to obtain
\begin{equation} \label{eqn:mass-on-k}
\frac{\|\omega\|_1}{|\omega(k)|} \le \sum_{t = 0}^k \binom{k}{t} + \frac{1}{48\cdot 4^k\cdot \sqrt{N} \log{N}} < 2^k + 1 < \frac{1}{2} \cdot 4^k.
\end{equation}
We calculate 
\begin{align*}
\frac{\|\omega\|_1}{2} - \sum_{t \in E_-} |\omega(t)| &= \frac{1}{2}\left(\sum_{t:\omega(t)<0}|\omega(t)|+\sum_{t:\omega(t)>0}|\omega(t)| \right)-\sum_{t\in E_-}(-\omega(t)) \\
&=\sum_{t : \omega(t) < 0} (-\omega(t)) - \sum_{t \in E_-} (-\omega(t))  \tag*{since $\langle \omega, \mathbf{1} \rangle = 0$, so $\sum_{t:\omega(t)<0}|\omega(t)|=\sum_{t:\omega(t)>0}|\omega(t)|$}\\
&= \sum_{t \colon \omega(t) < 0, t \geq k} (-\omega(t)) \tag*{since $E_{-}= \{t\colon \omega(t) < 0, t < k\}$}\\
&\ge -\omega(k).
\end{align*}
Rearranging and applying the bound in Equation~\eqref{eqn:mass-on-k},
\[\sum_{t \in E_{-}} |\omega(t)| \le \left(\frac{1}{2} + \frac{\omega(k)}{\|\omega\|_1}\right) \cdot \|\omega\|_1 \le \left(\frac{1}{2} - 2 \cdot 4^{-k}\right) \cdot \|\omega\|_1.\]
\end{proof}

\end{document}